\documentclass[letterpaper,USenglish,cleveref,autoref,numberwithinsect,thm-restate]{lipics-v2021}

\pdfoutput=1 
\hideLIPIcs  

\title{Tight Guarantees for Cut-Relative Survivable Network Design via a Decomposition Technique} 

\titlerunning{Tight Guarantees for Cut-Relative Survivable Network Design} 

\author{Nikhil Kumar}{University of Waterloo}{nikhil.kumar2@uwaterloo.ca}{}{}
\author{JJ Nan}{University of Waterloo}{jnan@uwaterloo.ca}{}{}
\author{Chaitanya Swamy}{University of Waterloo}{cswamy@uwaterloo.ca}{https://orcid.org/0000-0003-1108-7941}{}

\authorrunning{N. Kumar, J. Nan, and C. Swamy} 



\relatedversion{A preliminary version is to appear in the Proceedings of the
  European Symposium on Algorithms (ESA) 2025.} 

\funding{Supported in part by NSERC grant 327620-09.}

\nolinenumbers 

\usepackage{amssymb}
\usepackage{amsmath}
\usepackage{xspace,multirow,comment,bbm}
\usepackage[shortlabels]{enumitem}
\usepackage{calc}
\usepackage{amsthm}
\usepackage{graphicx}
\usepackage[ruled, linesnumbered, vlined, nokwfunc]{algorithm2e}
\usepackage{caption,makecell}

\makeatletter
\newcommand{\clonelabel}[2]{\@bsphack
  \expandafter\ifx\csname r@#2\endcsname\relax
  \else\protected@write\@auxout{}{\string\newlabel{#1}%
    {\csname r@#2\endcsname}}%
  \fi
  \expandafter\ifx\csname r@#2@cref\endcsname\relax
  \else\protected@write\@auxout{}{\string\newlabel{#1@cref}%
    {\csname r@#2@cref\endcsname}}%
  \fi
  \@esphack}
\makeatother

\newcommand{\np}{{\em NP}\xspace}
\newcommand{\nphard}{\np-hard\xspace} 

\newcommand{\apx}{{\em APX}\xspace}
\newcommand{\apxhard}{\apx-hard\xspace}

\DeclareMathOperator{\argmax}{argmax}

\DeclareMathOperator{\defc}{def}
\DeclareMathOperator{\wdef}{wdef}

\newtheorem{fact}[theorem]{Fact}

\newenvironment{proofof}[1]{\begin{proof}[Proof of {#1}]}{\end{proof}}

\newcommand{\R}{\ensuremath{\mathbb R}}
\newcommand{\Z}{\ensuremath{\mathbb Z}}

\newcommand{\C}{\ensuremath{\mathcal{C}}}

\newcommand{\Lc}{\ensuremath{\mathcal L}}

\newcommand{\OPT}{\ensuremath{\mathit{OPT}}}

\newcommand{\es}{\ensuremath{\emptyset}}
\newcommand{\assign}{\ensuremath{\leftarrow}}

\newcommand{\sse}{\subseteq}

\newcommand{\hx}{\ensuremath{\widehat x}}

\newcommand{\bx}{\ensuremath{\overline x}}

\newcommand{\bA}{\ensuremath{\overline A}}
\newcommand{\bS}{\ensuremath{\overline S}}

\newcommand{\ld}{\ensuremath{\lambda}}

\newcommand{\al}{\ensuremath{\alpha}}

\newcommand{\dt}{\ensuremath{\delta}}

\newcommand{\sndp}{\ensuremath{\mathsf{SNDP}}\xspace}
\newcommand{\kecss}{\ensuremath{k\text{-}\mathsf{ECSS}}\xspace}
\newcommand{\kefts}{\ensuremath{k\text{-}\mathsf{EFTS}}\xspace}
\newcommand{\fndp}[1][f]{\ensuremath{{#1}\text{-}\mathsf{NDP}}\xspace}
\newcommand{\crsndp}{\ensuremath{\mathsf{CR}\text{-}\mathsf{SNDP}}\xspace}
\newcommand{\crfndp}[1][f]{\ensuremath{\mathsf{CR}\text{-}{#1}\text{-}\mathsf{NDP}}\xspace}
\newcommand{\lpname}[1][f,G]{\ensuremath{(\text{\textnormal{CRLP}}_{#1})}\xspace}
\newcommand{\fndlp}[1][f,G]{\ensuremath{(\text{\textnormal{NDP}}_{#1})}\xspace}
\newcommand{\scut}{A}
\newcommand{\scutc}{\bA}
\newcommand{\gs}[1][S]{\ensuremath{G[{#1}]}}
\newcommand{\sym}{\ensuremath{\mathsf{sym}}}

\newcommand{\crndpalg}{{\sc CRNDP-Alg}\xspace}
\newcommand{\scol}{\C}
\newcommand{\alg}{\ensuremath{\mathsf{Alg}}\xspace}
\newcommand{\thresh}{\tau}
\newcommand{\prop}{\pi}
\newcommand{\grace}{\ensuremath{\mathsf{grace}}\xspace}
\newcommand{\cdfunc}{h}
\newcommand{\smcuts}{Z}
\newcommand{\nf}{\gamma}
\newcommand{\newf}{\kappa}
\newcommand{\spart}{C}
\newcommand{\eqclass}{\ensuremath{\mathsf{EQ}}\xspace}

\SetAlgoProcName{Algorithm}{Algorithm}
\SetFuncSty{textsc}
\SetAlCapNameSty{textsc}

\begin{document}

\maketitle

\begin{abstract}
In the classical {\em survivable-network-design problem} (\sndp), we are given an
undirected graph \( G = (V, E) \), non-negative edge costs, and 
some $k$ tuples $(s_i,t_i,r_i)$, where $s_i,t_i\in V$ and $r_i\in\Z_+$.
The objective is to find a
minimum-cost subset \( H \subseteq E \) such that each $s_i$-$t_i$ pair remains connected
even after the failure of any $r_i-1$ 
edges. It is well-known
that \sndp can be equivalently modeled using a weakly-supermodular 
{\em cut-requirement function} \( f \), where the objective is to find the minimum-cost
subset of edges that picks at least \( f(S) \) edges across every cut \( S \subseteq V \). 

Recently, motivated by fault-tolerance in graph spanners, Dinitz, Koranteng, and Kortsartz
proposed a variant of \sndp that enforces a {\em relative} level of fault tolerance with
respect to \( G \). Even if a feasible \sndp-solution may not exist due to \( G \)
lacking the required fault-tolerance, the goal is to find a solution $H$ that is at least
as fault-tolerant as \( G \) itself. 
They formalize the latter condition in terms of paths and fault-sets, which gives rise to  
\emph{path-relative \sndp} (which they call relative \sndp). 
Along these lines, we introduce a new model of relative network design, called 
\emph{cut-relative \sndp} (\crsndp), where the goal is to select a minimum-cost subset of
edges that satisfies the given (weakly-supermodular)
cut-requirement function to the maximum extent possible, i.e., by picking \( \min\{ f(S),
|\delta_G(S)| \} \) edges across every cut \( S\sse V \). 

Unlike \sndp, the cut-relative and path-relative versions of \sndp are not equivalent. 
The resulting cut-requirement function for \crsndp (as also path-relative \sndp) is not
weakly supermodular, and extreme-point solutions to the natural LP-relaxation need not
correspond to a laminar family of tight cut constraints.
Consequently, standard techniques cannot be used directly to design approximation
algorithms for this problem. We develop a {\em novel decomposition technique} to circumvent this
difficulty and use it to give a 
{\em tight $2$-approximation algorithm for \crsndp}. 
We also show some new hardness results for these relative-\sndp problems. 
\end{abstract}

\section{Introduction} \label{intro}

In the classical {\em survivable-network-design problem} (\sndp), we are given an
undirected graph $G=(V,E)$, non-negative edge costs $\{c_e\}_{e\in E}$, and some $k$
source-sink pairs $s_i,t_i\in V$ with requirements $r_i\in\Z_+$, for $i=1,\ldots,k$. The
goal is to find a minimum-cost set $H\sse E$ of edges such that there are at least $r_i$
edge-disjoint $s_i$-$t_i$ paths in $H$, for all $i\in[k]$. We overload notation and
use $H$ to denote both the edge-set, and the corresponding subgraph of $G$.

\sndp imposes an absolute level of fault-tolerance in a solution subgraph $H$, 
by ensuring that for any bounded-size {\em fault-set} $F\sse E$ of edges that may fail,
the graph $H-F$ retains some connectivity properties.
More precisely, due to Menger's theorem, or the max-flow min-cut theorem, the feasibility
condition on $H$ can be stated equivalently as follows.
For every $i\in[k]$, and every fault-set $F\sse E$ with $|F|<r_i$: 
\begin{enumerate}[label=(S\arabic*), topsep=2pt, noitemsep, leftmargin=*]
\item (Path-version) there is an $s_i$-$t_i$ path in $H-F$; \label{pathsndp}
\item (Cut-version) 
for every $s_i$-$t_i$ cut $S\sse V$ (i.e., $|S\cap\{s_i,t_i\}|=1$), 
we have $\dt_{H-F}(S)\neq\es$, where $\dt_{H-F}(S):=\dt(S)\cap(H-F)$. \label{cutsndp}
\end{enumerate}

Recently, motivated by work on fault-tolerance in graph spanners, Dinitz, Koranteng, and 
Kortsartz~\cite{DinitzKK22}, proposed a variant of \sndp that aims to impose a  
{\em relative} level of fault-tolerance with respect to the graph $G$. 
The idea is that even if the \sndp instance in infeasible, because $G$ itself
does not possess the required level of fault-tolerance, one should not have to completely
abandon the goal of fault-tolerance: 
it is still meaningful and reasonable to seek a solution that is ``as fault-tolerant
as $G$'', and in this sense is fault-tolerant relative to $G$. 

To formalize this, it is useful to consider the definition of \sndp in terms
of fault-sets. 
Roughly speaking, we would like to capture that $H$ is a feasible solution if
\begin{equation}
\text{for every fault-set $F$ (valid for \sndp), \quad
$H-F$ and $G-F$ have similar connectivity.}
\tag{*} \label{infdefn}
\end{equation}
To elaborate, in \sndp, if there is {\em even one} fault-set $F'$ under which $G-F'$ fails
to have the required connectivity, i.e., 
$|F'|<r_i$ and $F'\supseteq\dt_G(S)$ for some $s_i$-$t_i$ cut $S$, 
then ``all bets are off''; we declare that the instance in infeasible.
But declaring infeasibility here feels unsatisfactory because we allow one specific
``problematic'' fault-set $F'$ 
to block us from obtaining any kind of fault-tolerance.  
In contrast, \eqref{infdefn} aims to provide
a {\em per-fault-set guarantee}, by 
asking for a solution $H$ that functions as well as $G$ in terms of connectivity under the  
failure of any fault-set $F$ (considered for \sndp). 
We can formalize ``$H-F$ and $G-F$ have similar connectivity'' in two ways, via paths or
cuts, and this gives rise to the following two problem definitions.

\begin{definition} \label{relsndp}
Let $\bigl(G=(V,E),\{c_e\}_{e\in E},\{s_i,t_i,r_i\}_{i\in[k]}\bigr)$ be an \sndp
instance. 
\begin{enumerate}[label=\textnormal{(R\arabic*)}, topsep=2pt, noitemsep, leftmargin=*]
\item {\em Path-relative \sndp}: $H\sse E$ is a feasible solution
if for every $i\in[k]$ and $F\sse E$ with $|F|<r_i$, \ \ 
$G-F$ has an $s_i$-$t_i$ path \ \ $\implies$ \ \ $H-F$ has an $s_i$-$t_i$ path.
\label{prsndp}

\item {\em Cut-relative \sndp} (\crsndp): $H\sse E$ is a feasible solution
if for every $i\in[k]$ and $F\sse E$ with $|F|<r_i$, 
for every $s_i$-$t_i$ cut $S$, \quad
$\dt_{G-F}(S)\neq\es\ \implies\ \dt_{H-F}(S)\neq\es$.
\label{crsndp-def}
\end{enumerate}
The goal in both problems is to find a minimum-cost feasible solution.
\end{definition}

Perhaps surprisingly, path-relative \sndp and cut-relative \sndp are 
{\em not equivalent}, even when the \sndp instance 
involves a single $s$-$t$ pair. It is not hard to show that if $H$ is feasible for
\crsndp, it is also feasible for path-relative SNDP%
\footnote{Fix any $i\in[k]$ and $F\sse E$ with $|F|<r_i$. If $G-F$ has an
$s_i$-$t_i$ path, then $\dt_{G-F}(S)\neq\es$ for every $s_i$-$t_i$ cut. So since $H$ is
feasible for cut-relative \sndp, we have $\dt_{H-F}(S)\neq\es$ for every $s_i$-$t_i$ cut,
which implies that $H-F$ has an $s_i$-$t_i$ path.}
but the converse fails to hold; see Appendix~\ref{append-nonequiv} for an example.

Path-relative \sndp was defined by Dinitz et al.~\cite{DinitzKK22}, who referred to this
problem simply as relative \sndp.
They considered (among other problems) the path-relative version of  
$k$-edge connected subgraph (\kecss), which is the special case of \sndp where we have an
$s_i$-$t_i$ pair with $r_i=k$ for every pair of nodes. (In this case, the path-relative
and cut-relative versions do turn out to be equivalent~\cite{DinitzKK22}; see also
Lemma~\ref{kecss-equiv}.) 
They called the resulting path-relative problem, {\em $k$-edge fault-tolerant subgraph}
(\kefts), and observed that  
the feasibility condition for \kefts can be equivalently stated as: $H$ is feasible iff
$|\dt_H(S)|\geq g(S):=\min\bigl\{k,|\dt_G(S)|\bigr\}$ for all $S\sse V$. 
The function
$g:2^V\mapsto\Z$ is called a {\em cut-requirement function}, as it stipulates the (minimum)
number of edges across any cut in any feasible solution. 
Cut-requirement functions constitute a very versatile framework for specifying
network-design problems, where given
a cut-requirement function $f:2^V\mapsto\Z$, the
corresponding {\em $f$-network-design problem} (\fndp), also called the 
{\em $f$-connectivity problem}, 
is to find a minimum-cost set $H\sse E$ such that 
$|\dt_H(S)|\geq f(S)$ for all $S\sse V$. 
For instance, it is easy to see that \sndp corresponds to \fndp 
for the cut-requirement function $f^\sndp$, defined by
$f^\sndp(S):=\max\,\{r_i: |S\cap\{s_i,t_i\}|=1\}$ 
(and \kecss is $\fndp[{f^{\kecss}}]$ where $f^{\kecss}(S):=k$ for all 
$\es\neq S\subsetneq V$); 
various other network-design problems, such as the $T$-join problem, point-to-point
connection problem etc., can be modeled using suitable cut-requirement functions 
(see, e.g.,~\cite{GoemansW95,GoemansW96}). 

Cut-relative \sndp is a problem that we introduce in this paper. 
Its formulation in terms of cuts is the problem one 
obtains when we 
replace $f^{\kecss}$ in the cut-based formulation of
\kefts 
by the cut-requirement function $f^\sndp$ for (general) \sndp: that is, 
(similar to \kefts), we can say that $H\sse E$ is feasible for \crsndp iff 
$|\dt_H(S)|\geq g^{\crsndp}(S):=\min\bigl\{f^\sndp(S),|\dt_G(S)|\bigr\}$ for all $S\sse V$  
(we show this easy equivalence in Section~\ref{prelim}); in other words, \crsndp
corresponds to $\fndp[{g^{\crsndp}}]$. 
To gain some intuition for cut-relative \sndp, and contrast it with path-relative \sndp,  
consider \sndp with a single $s$-$t$ pair and requirement $r$, for simplicity.
A feasible solution $H$ to path-relative \sndp offers a per-fault-set guarantee as
encapsulated by~\ref{prsndp}.
But, for a given fault-set $F$ with $|F|<r$, we get a weak fault-tolerance
guarantee in terms of $s$-$t$ cuts: if $\delta_{G-F}(S)=\emptyset$ for \emph{even one}
$s$-$t$ cut $S$ then there is no requirement on $H$ for this fault-set $F$. Cut-relative
\sndp offers a per-fault-set \emph{and per-cut guarantee}, since for (every $F$ with
$|F|<r$ and) every $s$-$t$ cut $S$, we require that 
$\delta_{H-F}(S)\neq\emptyset$ if $\delta_{G-F}(S)\neq\emptyset$.

Along the lines of \crsndp, we can easily extend the cut-relative model to capture, more
broadly, 
the cut-relative version of any network-design problem specified by a 
cut-requirement function:   
given a cut-requirement function $f$, 
in the corresponding {\em cut-relative $f$-network-design problem}
(\crfndp), 
we seek a min-cost $H\sse E$ such that 
$|\dt_H(S)|\geq g^{\crfndp}(S):=\min\bigl\{f(S),|\dt_G(S)|\bigr\}$ for all $S\sse V$; 
that is, we seek to satisfy the cut-requirement function $f$ to the maximum extent
possible.
We refer to $f$ as the {\em base cut-requirement function}, to distinguish it from the
cut-requirement function $g^{\crfndp}$ that defines \crfndp.

\vspace*{-1ex}
\subparagraph*{Modeling power.}
We arrived at \crfndp, which can be seen as a fault-tolerant
cut-covering problem, 
as a technically natural extension of \kefts. We show below
that \crfndp offers a surprising amount of modeling power, 
and, in particular, allows one to capture stronger forms of relative fault-tolerance 
compared to \kefts.  
In \kefts, there is a sharp relative-fault-tolerance threshold at $k$: if $H$ is feasible,
then $G-F$ and $H-F$ have the same components whenever $|F|<k$, but there are
no guarantees for larger fault-sets. 
With \crfndp, 
one can capture 
a weaker guarantee also for other fault-sets, which
allows for a more {\em graceful degradation} of relative fault-tolerance as $|F|$
increases.   

\begin{example} \label{graceful}
One way of modeling graceful relative-fault-tolerance degradation is as follows: 
given a {\em non-increasing function} $\thresh:\Z_+\mapsto\Z_+$, 
we seek a min-cost subgraph $H$ such that for {\em every} fault-set $F$,  
$G-F$ and $H-F$ have exactly the same components with at most $\thresh(|F|)$ nodes. 
If $G$ models a communication network where 
connected nodes can communicate with each other, then 
this yields the desirable guarantee that,
{\em post-faults}, if a node can communicate with more than $\thresh(|F|)$ nodes in
$G$ then the same holds for $H$; 
if this number is at most $\thresh(|F|)$, then post-faults, it can communicate with the
same nodes in $G$ and $H$.
Note that if $\thresh(k-1)\geq n=|V|$, then $H$ is feasible for \kefts; 
by adding more ``levels'' to $\thresh$, we can obtain fault-tolerance guarantees for 
larger fault-sets. 

In Section~\ref{modeling}, we show that this problem can be modeled as \crfndp 
with the weakly-supermodular function $f=f^{\grace}$, where
$f^{\grace}(S):=\min\,\{\ell: \thresh(\ell)<|S|\}$, for $S\neq\es$. 
\end{example} 

\begin{example} \label{gengraceful}
Generalizing Example~\ref{graceful} substantially, suppose along with $\thresh$, 
we have a {\em monotone} function $\prop:2^V\mapsto\R_+$, i.e., $\prop(T)\leq\prop(S)$
if $T\sse S$.
We now seek a subgraph $H$ such that for every fault-set $F$ and $S\sse V$ with 
$\prop(S)\leq\thresh(|F|)$, we have that 
$S$ is (the node-set of) a component of $G-F$ iff it is a component of $H-F$.
Similar to Example~\ref{graceful}, 
we can model this as \crfndp by defining
$f(S):=\min\,\{\ell: \thresh(\ell)<\prop(S)\}$, for $S\neq\es$ 
(Theorem~\ref{gengraceful-thm}). 

This setup yields fault-tolerance degradation depending on monotone
properties of components other than just their size, which creates a rich space of
problems. 
For example, suppose $\prop(S)=\max_{u,v\in S}\min_{\text{$u$-$v$ paths $P$ in $G$}}|P|$
is the weak-diameter of $S$. 
Setting $\thresh(\ell)=n$ if $\ell<k$, and some $t<n$ for larger $\ell$,  
we seek a \kefts 
solution $H$ such that for $F\sse E$ with $|F|\geq k$, 
$G-F$ and $H-F$ have the same components of weak-diameter at most $t$. 
\end{example}

\subsection{Our contributions} \label{contrib} 
We introduce cut-relative network-design problems and develop strong approximation
guarantees for these problems.  
{\em We obtain an approximation guarantee of $2$ for cut-relative \sndp},
which, notably, 
matches the best-known approximation factor for \sndp. 
Our guarantee applies more generally to \crfndp, whenever the base cut-requirement
function $f$ satisfies certain properties 
and the natural LP-relaxation of the problem \ref{lp} can be solved efficiently.
Our guarantee is relative to the optimal value of this LP,
{\em and is tight in that it matches the integrality gap of this LP.}

We also show that even in the simplest \sndp setting with only one $s$-$t$ pair
(wherein \sndp is polytime solvable), cut-relative \sndp and path-relative \sndp 
capture \kecss as a special case, and are thus \apxhard
(Section~\ref{hardness}). Previously, even \nphard{}ness of path-relative \sndp in the
$s$-$t$ case (which is studied 
by~\cite{DinitzKK22,DinitzKKN23}) was not known.

\vspace*{-1ex}
\subparagraph*{Technical contributions and overview.}
Technically, our main contribution is to show that we can obtain such a strong
guarantee {\em despite the fact that 
the cut-requirement function defining cut-relative network-design (i.e., $g^{\crfndp}$)
is not weakly supermodular}, which is the key property that drives the
$2$-approximation algorithm for \sndp.  
To elaborate, there is a celebrated $2$-approximation algorithm for \sndp due to
Jain~\cite{Jain99} based on iterative rounding. This {\em crucially} exploits the fact
that the underlying cut-requirement function $f^\sndp$ is {\em weakly supermodular}: for
any two node-sets $A,B\sse V$, we have: 
\[
f^\sndp(A)+f^\sndp(B)\leq\max\bigl\{f^\sndp(A\cap B)+f^\sndp(A\cup B),\ f^\sndp(A-B)+f^\sndp(B-A)\}.
\]
This property allows one to argue that given an LP-solution $x$, any two {\em tight} sets
that cross---i.e., sets $A, B$ for which $x(S)=f^\sndp(S)$ holds for $S\in\{A,B\}$ and 
$A\cap B, A-B, B-A$ are all non-empty---can be ``uncrossed'', and thereby an
extreme-point solution $\hx$ to the LP can be defined via a {\em laminar family} of tight
cut constraints. By a laminar family, we mean that any two sets $A, B$ in the family
satisfy $A\cap B=\es$ or $A\sse B$, and uncrossing two crossing tight sets means
that they can be replaced by an equivalent laminar family of tight sets.
Jain's seminal contribution was to show that 
given such a laminar family defining an extreme point, there always exists an edge $e$ for
which $\hx_e\geq\frac{1}{2}$; picking such an edge and iterating then yields the
$2$-approximation. This technique of uncrossing tight cut-constraints to obtain a
laminar family of tight cut-constraints has proved to be quite versatile,
becoming a staple technique 
that has been leveraged to obtain strong guarantees in various network-design settings
such as, most notably, network-design problems with degree
constraints~\cite{Goemans06,LauNSS09,SinghL07,BansalKN09,BansalKKNP13,KleinO23}.   
More generally, this technique of uncrossing to obtain a structured family of tight
constraints has been utilized for various combinatorial-optimization problems; see,
e.g.,~\cite{LauRS11}. 

As mentioned earlier, the main complication with \crfndp is that the
underlying cut-requirement function $g^{\crfndp}$ is usually {\em not} weakly-supermodular, 
even if the base cut-requirement function $f$ is. In particular, the cut-requirement
function $g^{\crsndp}$ defining cut-relative \sndp is not weakly-supermodular. This
difficulty was also noted by Dinitz et al.~\cite{DinitzKK22} when they considered $k$-edge
fault-tolerant subgraph (\kefts, or path- or cut- relative \kecss in our terminology). For
this special case, they identify a notion called local weak-supermodularity that is
satisfied by the cut-requirement function $\min\{k,|\dt_G(S)\}$ for \kefts, and
observed that Jain's machinery can be used as is if this property holds; that is, one can
uncross tight sets to obtain a laminar family, and hence obtain a $2$-approximation for
\kefts proceeding as in Jain's algorithm for \sndp.

However, for general \crsndp, such an approach breaks down, because, as we show in
Appendix~\ref{append-nolaminar}, one can construct quite simple \crsndp instances where an
(optimal) extreme point 
{\em cannot be defined by any laminar family of tight cut constraints}.    
In particular, this implies that the cut-requirement function $g^{\crfndp}$ is (in
general) not locally weakly-supermodular (even when $f$ is weakly-supermodular); moreover,  
it does not satisfy any property that allows for the uncrossing of tight sets. 
The upshot here is that the standard technique of uncrossing tight sets to obtain a
laminar family does not work, and one needs new ideas to deal with cut-relative network
design. 

The chief technical novelty of our work is to show how to overcome the impediment posed by 
the lack of any such nice structural property for $g^{\crfndp}$, 
via a novel {\em decomposition technique} that 
allows one to {\em reduce} cut-relative \fndp with a weakly supermodular 
{\em base cut-requirement function} to (standard) \fndp with (different) weakly
supermodular cut-requirement functions (see Theorem~\ref{redn-cor}).  

We introduce a little notation to describe this.
Let $\crfndp[(h,D)]$ denote \crfndp 
on graph $D$ with base  cut-requirement function $h$. 
For $S\sse V$, define its
{\em deficiency} to be $f(S)-|\dt_G(S)|$; we say that $S$ is a {\em small cut} if its
deficiency is positive (i.e., $|\dt_G(S)|<f(S)$). 
We use $\gs$ 
to denote the subgraph induced by $S$.
Let $\bS$ denote $V-S$. 

Assume that the underlying base cut-requirement function $f$ is
weakly supermodular (as is the case with \sndp). 
Observe first that if there is no small cut then $\crfndp$ is the same as $\fndp$.
Our key technical insight is that, otherwise,
we can utilize a novel splitting operation to simplify our problem. 
We show (see Theorem~\ref{main-redn}) that if we pick a 
small cut $\scut\sse V$ with {\em maximum deficiency},
then one can define suitable weakly-supermodular functions 
$f_\scut:2^\scut\mapsto\Z$ and $f_{\scutc}:2^{\scutc}\mapsto\Z$ such that {\em any
solution to the 
original \crfndp-instance yields solutions to $\crfndp[{(f_\scut,\gs[\scut])}]$ 
and $\crfndp[{(f_{\scutc},\gs[\scutc])}]$, 
and vice-versa}. 
The intuition here is as follows. Since $\scut$ is a small cut, any feasible \crfndp
solution must include all edges in $\dt_G(\scut)$. The functions $f_\scut$ and
$f_{\scutc}$ are essentially 
the restrictions of $f$ to $\scut$ and $\scutc$ respectively, taking this into account. 
Consequently, moving to $\crfndp[{(f_\scut,\gs[\scut])}]$ and
$\crfndp[{(f_{\scutc},\gs[\scutc])}]$ does not really 
impact the constraints for cuts that do not cross $\scut$;
but we do eliminate the constraints for crossing cuts. 
This is the simplification that we obtain via the splitting operation,
and one needs to argue that this does not hurt feasibility; that is, the
constraints for such cuts are still satisfied when we take $\dt_G(\scut)$ along with the
$\crfndp[{(f_\scut,\gs[\scut])}]$ and $\crfndp[{(f_{\scutc},\gs[\scutc])}]$ solutions.  
Here we exploit the weak-supermodularity of $f$, which implies 
that the
deficiency function $\defc(S):=f(S)-|\dt_G(S)|$ is also weakly supermodular. 
One can use this, and the fact that $\scut$ has maximum deficiency, to bound the
deficiency of any crossing set $T$ in terms of the deficiency of non-crossing sets, and
hence show that the cut-constraint for $T$ is satisfied.

By recursively applying the above splitting operation to $\scut$ and $\scutc$, we end up
with a partition $V_1,\ldots,V_k$ of $V$ and weakly-supermodular functions
$f_i:2^{V_i}\mapsto\Z$ for all $i\in[k]$ such that, for each $i\in[k]$, the graph $\gs[V_i]$
contains no small cut with respect to $f_i$. It follows that 
solving $\fndp[f_i]$ on the graph $\gs[V_i]$ for all $i\in[k]$ yields a solution to the
original \crfndp instance (Theorem~\ref{redn-cor}).  

\medskip
While we can obtain this decomposition efficiently, assuming that one has suitable
algorithmic primitives involving the base cut-requirement function $f$ (see
Section~\ref{decomp-comp}), in fact we only
need this decomposition in the {\em analysis}: 
in order to obtain the stated
LP-relative $2$-approximation algorithm, we only need to be able to find an extreme-point
optimal solution to the natural LP-relaxation \ref{crlp} for \crfndp (also after fixing
some variables to $1$). This is because the translation between \crfndp-solutions
for $G$, and $\fndp[f_i]$-solutions on $\gs[V_i]$ for all $i\in[k]$, also applies to
{\em fractional solutions}, 
and implies that any
extreme-point solution to \ref{crlp} yields extreme-point solutions to the
LP-relaxations of $\fndp[f_i]$ on $\gs[V_i]$ for all $i\in[k]$ (parts (b) and (c) of
Theorem~\ref{redn-cor}); hence, by Jain's result,
we can always find some fractional edge of value at least $\frac{1}{2}$, round its value
to $1$, and iterate. For \sndp, we can 
find an extreme-point optimal solution efficiently because we argue that one can devise an 
efficient separation oracle for \ref{crlp} (Lemma~\ref{crlp-solve}). 
This yields a $2$-approximation algorithm for \crsndp, and more
generally \crfndp, whenever $f$ is weakly supermodular and \ref{crlp} can be solved
efficiently (see Algorithm \crndpalg in Section~\ref{algresults}).

\subparagraph*{Other related work.}
Network design is a fairly broad research topic, with a vast amount of literature. 
We limit ourselves to the work that is most closely connected to our work.

As mentioned earlier, path-relative \sndp was introduced by Dinitz et
al.~\cite{DinitzKK22}. 
In addition to \kefts, which is the path-relative version of \kecss, 
they study the path-relative versions of two other special
cases of \sndp, for which they developed constant-factor approximation algorithms: 
(1) \sndp with $r_i\leq 2$ for all $i$;
(2) one $s_i$-$t_i$ pair with $r_i\leq 3$.
Both results were subsequently extended by~\cite{DinitzKKN23} who gave a 
$2$-approximation for \sndp with $r_i\leq 3$ for all $i$, and
a $2^{O(k^2)}$-approximation when there is one $s_i$-$t_i$ pair with $r_i=k$.
The only result known for general path-relative \sndp is due to~\cite{ChekuriJ23}, who
(among other results) devise an $O(k^2\log^2n\log\log n)$-approximation algorithm in
$n^{O(k)}$ time. 
As noted earlier, on the hardness side, prior to our work, it was not known if the setting
with one $s_i$-$t_i$ pair is even \nphard.

An influential paper by Goemans and Williamson~\cite{GoemansW95} (see also the
survey~\cite{GoemansW96}), which built upon the 
work of~\cite{AgrawalKR95} for the Steiner forest problem, developed the primal-dual method
for network-design problems. Their work also popularized the use of cut-requirement
functions to specify network-design problems, by showing how this framework can be used to
capture a slew of problems, and how their primal-dual framework leads to a
$2$-approximation algorithm for $\{0,1\}$-cut-requirement functions satisfying some
properties. This work was extended to handle general integer-valued cut-requirement
functions, which captures \sndp, in~\cite{WilliamsonGMV95,GoemansGPSTW94}, which led to an
$O(\log\max_i r_i)$-approximation for \sndp. The seminal work of Jain~\cite{Jain99} later
improved this to a $2$-approximation, and in doing so, introduced the powerful technique of
{\em iterative rounding}. Later work of~\cite{LauNSS09,SinghL07} on degree-bounded network
design added another ingredient to iterative rounding, namely iteratively
dropping some constraints. This paradigm of {\em iterative rounding and relaxation} has
proved to be an extremely versatile and powerful tool in developing algorithms (and
structural results) for
network-design problems, and more generally in combinatorial optimization;
see~\cite{LauRS11} for an extensive study.

Iterative rounding and relaxation derives its power from the fact that an extreme-point
solution to an LP-relaxation of the problem can be defined using
a structured family of tight constraints, such as, often a laminar family of cut
constraints in the case of network-design problems. As noted earlier, this key property
does not hold for cut-relative \sndp (and path-relative \sndp). Recently,
various works have considered network-design problems that give rise to cut constraints
that are less structured and not quite uncrossable;
see~\cite{Bansal25,BansalCGI24,BoydCHI24} and the references therein. Our work can be seen
as furthering this research direction.

\section{Preliminaries and notation} \label{prelim}
Recall that we are given an undirected graph $G=(V,E)$ with nonnegative edge costs
$\{c_e\}_{e\in E}$. For a subset $S\sse V$ and subset $Z\sse E$, which we will
interchangeably view as the subgraph $(V,Z)$ of $G$, we use $\dt_Z(S)$ to denote
$\dt_G(S)\cap Z$, i.e., the edges of $Z$ on the boundary of $S$.
Recall that $\bS$ denotes $V-S$, and $\gs$ denotes the subgraph induced by $S$.
A cut-requirement function on $G$ is a function $f:2^V\mapsto\Z$.
(We allow $f(S)$ to be negative chiefly for notational simplicity; 
this does not affect anything, as we will only examine sets for which $f(S)\geq 0$.)
Recall the following network-design problems.
\begin{enumerate}[label=$\bullet$, topsep=0.4ex, itemsep=0.2ex, leftmargin=*]
\item {\em $f$-network-design problem} (\fndp): find a min-cost edge-set $H$ such that
$|\dt_H(S)|\geq f(S)$ for all $S\sse V$.

\item {\em Cut-relative network-design problem} (\crfndp): find a min-cost edge-set $H$
such that $|\dt_H(S)|\geq g^{\crfndp}(S):=\min\bigl\{f(S),|\dt_G(S)|\bigr\}$ for all 
$S\sse V$. This corresponds to $\fndp[{g^{\crfndp}}]$. We call $f$ the base cut-requirement
function. 

\item {\em Survivable network design problem} (\sndp) is a special case of \fndp. 
The input also specifies $k$ tuples $\{(s_i,t_i,r_i)\}_{i\in[k]}$, where 
$s_i,t_i\in V$ and $r_i\in\Z_+$ for all $i\in[k]$. Defining 
$f^\sndp(S):=\max\,\{r_i: |S\cap\{s_i,t_i\}|=1\}$ for all $S\sse V$, 
we obtain that \sndp is the $f^\sndp$-network-design problem.

\item {\em Cut-relative \sndp} (\crsndp): 
The input here is the same as in \sndp. In Definition~\ref{relsndp} \ref{crsndp-def},
\crsndp was defined as: find a min-cost $H\sse E$ such that for every $i\in[k]$, every
$F\sse E$ with $|F|<r_i$, and every $s_i$-$t_i$ cut $S$, we have $\dt_{H-F}(S)\neq\es$
whenever $\dt_{G-F}(S)\neq\es$. Later in Section~\ref{intro}, we stated 
that this can be equivalently formulated as
\crfndp with base cut-requirement function $f^\sndp$. 
We now justify this statement.

Suppose that $H$ is feasible under \ref{crsndp-def} 
but $|\dt_H(S)|<\min\bigl\{f^\sndp(S),|\dt_G(S)|\bigr\}$ for some $S\sse V$. Suppose
$f^\sndp(S)=r_i$. Then taking $F=\dt_H(S)$, we have $|F|<r_i$, $\dt_{G-F}(S)\neq\es$, but
$\dt_{H-F}(S)=\es$, yielding a contradiction. So $H$ is feasible for $\crfndp[f^\sndp]$.
Conversely, suppose $H$ is feasible for $\crfndp[f^\sndp]$.
Consider any $i\in[k]$, any $F\sse E$ with $|F|<r_i$, and any $s_i$-$t_i$ cut $S$ with
$\dt_{G-F}(S)\neq\es$. Then, 
$|\dt_F(S)|<\min\bigl\{r_i,|\dt_G(S)|\bigr\}\leq\min\bigl\{f(S),|\dt_G(S)|\bigr\}\leq|\dt_H(S)|$; 
so $\dt_{H-F}(S)\neq\es$, showing that $H$ satisfies \ref{crsndp-def}.
\end{enumerate}

We will consider 
network-design problems defined on various
graphs and with various cut-requirement functions.  
Given a graph $D=(V_D,E_D)$ and cut-requirement function $h:2^{V_D}\mapsto\Z$, 
we use $\fndp[(h,D)]$ to denote \fndp on the graph $D$ with cut-requirement function $h$, 
and $\crfndp[(h,D)]$ to denote \crfndp on the graph $D$ with base cut-requirement function
$h$. 

\vspace*{-1ex}
\subparagraph*{Properties of cut-requirement functions.}
Let $f:2^V\mapsto\R$. We say that $f$ is {\em weakly supermodular} if for any two
node-sets $A,B\sse V$, we have 
$f(A)+f(B)\leq\max\bigl\{f(A\cap B)+f(A\cup B),\,f(A-B)+f(B-A)\bigr\}$. 
We say that $f$ is {\em symmetric} if $f(S)=f(V-S)$ for all $S\sse V$, and is 
{\em normalized} if $f(\es)=f(V)=0$.

Since we are working on undirected graphs, we 
can always replace any cut-requirement function $f$ by its symmetric, normalized version
$f^\sym$, where $f^\sym(S):=\max\{f(S),f(V-S)\}$ for all $\es\neq S\subsetneq V$ and
$f^\sym(\es)=f^\sym(V):=0$, without changing the resulting  
network-design problem(s), i.e., we have $\fndp[f]\equiv\fndp[f^\sym]$ and
$\crfndp[f]\equiv\crfndp[f^\sym]$. 

Lemma~\ref{wsupm-props} shows that weak-supermodularity is preserved under various
operations. This will be quite useful in the analysis of our decomposition technique. 
Part~\ref{wsupm-sym} of the lemma shows that weak supermodularity is preserved 
under symmetrization and normalization, 
so for any \fndp or \crfndp instance involving a weakly-supermodular
base cut-requirement function $f$, we may assume that $f$ is symmetric and normalized.  

\begin{lemma} \label{wsupm-props} \label{wsupm-prop}
Let $f:2^V\mapsto\R$ be weakly supermodular. Then
\begin{enumerate}[label=(\alph*), topsep=0.2ex, itemsep=0.1ex, leftmargin=*]
\item \label{wsupm-sym} \label{wsupm-norm}
$f^\sym$ is weakly supermodular.

\item \label{wsupm-res}
Let $w\in\R_+^E$. Define $g_1(S):=f(S)-w\bigl(\dt_G(S)\bigr)$ for all $S\sse V$.
Then $g_1$ is weakly supermodular, and $g_1$ is symmetric if $f$ is symmetric.
In particular, 
for any $Z\sse E$, the function $f(S)-|\dt_Z(S)|$
is weakly supermodular, and is symmetric if $f$ is symmetric.

\item \label{wsupm-proj}
Let $X\sse V$, and $\scol$ be a collection of subsets of $V-X$ closed under taking
set intersections, unions, differences, and complements with respect to $V-X$, 
i.e., if
$S\in\scol$ then $V-X-S\in\scol$, and 
$S,T\in\scol$ implies that $\{S\cap T, S\cup T, S-T, T-S\}\sse\scol$. 

Define the {\em projection of $f$ onto $X$ with respect to $\scol$} as the function 
$g_2:2^X\mapsto\R$ given by
$g_2(\es)=g_2(X):=0$, and $g_2(S):=\max_{T\in\scol}f(S\cup T)$ for all 
$\es\neq S\subsetneq X$. 
If $f$ is symmetric, then $g_2$ is weakly supermodular and symmetric. 
\end{enumerate}
\end{lemma}

\begin{proof} 
Part~\ref{wsupm-sym} follows from a simple case analysis;
part~\ref{wsupm-res} follows because $w\bigl(\dt_G(S)\bigr)$ is symmetric, submodular.
For part~\ref{wsupm-proj}, the closure properties of $\scol$ enable us to
transfer the weak-supermodularity of $f$ to the function $g_2$.

For part~\ref{wsupm-sym}, consider any $A,B\sse V$. 
If $A\sse B$ or $B\sse A$, 
then $f^\sym(A)+f^\sym(B)=f^\sym(A\cap B)+f^\sym(A\cup B)$;
if $A\cap B=\es$ or $V-(A\cup B)=\es$, then $\{A-B,B-A\}=\{A,B\}$ or
$\{A-B,B-A\}=\{V-B,V-A\}$, so  
$f^\sym(A)+f^\sym(B)=f^\sym(A-B)+f^\sym(B-A)$, where we also utilize the symmetry of $f^\sym$.

So assume that $A-B, B-A, A\cap B, V-(A\cup B)$ are all non-empty. 
We have $f^\sym(A)=f(X)$, where $X\in\{A,V-A\}$, and
$f^\sym(B)=f(Y)$, where $Y\in\{B,V-B\}$. So
\begin{equation*}
\begin{split}
f^\sym(A)+f^\sym(B) & \leq\max\bigl\{f(X\cap Y)+f(X\cup Y),\,f(X-Y)+f(Y-X)\bigr\} \\
& \leq\max\bigl\{f^\sym(X\cap Y)+f^\sym(X\cup Y),\ f^\sym(X-Y)+f^\sym(Y-X)\bigr\}.
\end{split}
\end{equation*}
We show that
\begin{equation}
\begin{split}
f^\sym(X\cap Y)+f^\sym(X\cup Y) & \in 
\bigl\{f^\sym(A\cap B)+f^\sym(A\cup B),\ f^\sym(A-B)+f^\sym(B-A)\bigr\}, \\
f^\sym(X-Y)+f^\sym(Y-X) & \in 
\bigl\{f^\sym(A\cap B)+f^\sym(A\cup B),\ f^\sym(A-B)+f^\sym(B-A)\bigr\}. 
\end{split}
\label{symineq}
\end{equation}
which will complete the proof of part~\ref{wsupm-sym}.
If $X=A$ and $Y=B$, then \eqref{symineq} clearly holds. So consider the remaining cases.
\begin{enumerate}[label=$\bullet$, topsep=0.1ex, itemsep=0.1ex, leftmargin=*]
\item $X=A$,\, $Y=V-B$. 
We have $X\cap Y=A-B$, and $X\cup Y=V-(B-A)$, so $f^\sym(X\cup Y)=f^\sym(B-A)$.
Also, $X-Y=A\cap B$, and $Y-X=V-(A\cup B)$, so $f^\sym(Y-X)=f^\sym(A\cup B)$. 
So \eqref{symineq} holds.

\item $X=V-A$,\, $Y=B$. This is completely symmetric to the above case.
We have $f^\sym(X\cap Y)=f^\sym(B-A)$ and $f^\sym(X\cup Y)=f^\sym(A-B)$.
Also, $f^\sym(X-Y)=f^\sym(A\cup B)$ and $f^\sym(Y-X)=f^\sym(A\cap B)$. 

\item $X=V-A$,\, $Y=V-B$. 
We have $X\cap Y=V-(A\cup B)$, so $f^\sym(X\cap Y)=f^\sym(A\cup B)$, and
$X\cup Y=V-(A\cap B)$, so $f^\sym(X\cup Y)=f^\sym(A\cap B)$.
Also, $X-Y=B-A$ and $Y-X=A-B$.
\end{enumerate}

\medskip
For part~\ref{wsupm-res}, since $w\bigl(\dt_G(S)\bigr)$ is a
symmetric, submodular function, we have 
\[
w\bigl(\dt_G(A)\bigr)+w\bigl(\dt_G(B)\bigr)\geq
\max\bigl\{w\bigl(\dt_G(A\cap B)\bigr)+w\bigl(\dt_G(A\cup B)\bigr),\ 
w\bigl(\dt_G(A-B)\bigr)+w\bigl(\dt_G(B-A)\bigr)\bigr\}
\]
for all $A,B\sse V$. Coupled with the weak-supermodularity of $f$, this yields that $g_1$
is weakly supermodular, and is symmetric if $f$ is symmetric.

\medskip
Finally, consider part~\ref{wsupm-proj}. We first argue that $g_2$ is symmetric.
If $A\in\{\es,X\}$, we have $g_2(A)=g_2(X-A)$ by definition.
Otherwise, if $g_2(A)=f(A\cup T_A)$, we have
$f(A\cup T_A)=f\bigl(V-(A\cup T_A)\bigr)=f\bigl((X-A)\cup (V-X-T_A)\bigr)\leq g_2(X-A)$
where the last inequality follows since $V-X-T_A\in\scol$. 
We also have $g_2(X-A)\leq g_2(A)$ by the same type of reasoning, and so $g_2(A)=g_2(X-A)$
for all $A\sse X$.

Now we show that $g_2$ is weakly supermodular. Consider $A,B\sse X$. 
If any of the sets 
$A-B, B-A, A\cap B, X-(A\cup B)$ are empty, then 
$g_2(A)+g_2(B)=g_2(A\cap B)+g_2(A\cup B)$ or $g_2(A)+g_2(B)=g_2(A-B)+g_2(B-A)$. 
The former holds when $A-B$ or $B-A$ is empty; the
latter holds when $A\cap B$ or $X-(A\cup B)$ is empty, where we also utilize the fact that
$g_2$ is symmetric. So assume otherwise.
Then, for every $S\in\{A, B, A\cap B, A\cup B, A-B, B-A\}$, we have
$g_2(S)=\max_{T\in\scol}f(S\cup T)$. 

Let $g_2(A)=f(A\cup T_A)$ and
$g_2(B)=f(B\cup T_B)$, where $T_A, T_B\in\scol$. Let $A'=A\cup T_A$ and 
$B'=B\cup T_B$. 
For any $S\in\{A'\cap B', A'\cup B', A'-B', B'-A'\}$, we can write $S=(S\cap X)\cup T$,
where $T\in\scol$ due to the closure properties of $\scol$.
It follows that $g_2(S\cap X)\geq f(S)$ for all $S\in\{A'\cap B',A'\cup B',A'-B',B'-A'\}$. 
So if $f(A')+f(B')\leq f(A'\cap B')+f(A'\cup B')$, we obtain that
$g_2(A)+g_2(B)\leq g_2(A\cap B)+g_2(A\cup B)$, and if $f(A')+f(B')\leq f(A'-B')+f(B'-A')$,
we obtain that $g_2(A)+g_@(B)\leq g_2(A-B)+g_2(B-A)$.
\end{proof}

The intuition behind the projection operation is as follows. Suppose 
$f$ is a cut-requirement function and we know that no edges are picked from
$E_G(V-X)\cup\dt_G(X)$. 
Then, roughly speaking, $g_2(S)$ captures the constraints that arise on
$S\sse X$ due to $f$ and $\scol$. 

\subsection{LP-relaxation for \boldmath \crfndp}
We consider the following natural LP-relaxation for \crfndp. From now on, we assume that
the cut-requirement function $f:2^V\mapsto\Z$ is weakly supermodular, symmetric, and
normalized. Given $\al\in\R^E$ and $F\sse E$, let $\al(F)$ denote $\sum_{e\in F}\al_e$.
\begin{alignat}{3}
\min &\quad & \sum_{e\in E}c_ex_e \tag*{\lpname} \label{lp} \\
\text{s.t.} \quad 
&& x\bigl(\dt_{G}(S)\bigr) & \geq\min\bigl\{f(S),|\dt_G(S)|\bigr\} 
\qquad && \forall S\sse V \label{cuts} \\ 
&& 0 \leq x_e &\leq 1 && \forall e\in E. \label{edges} 
\end{alignat}
\clonelabel{crlp}{lp}
\noindent
The LP-relaxation for \fndp is similar, with the cut-constraints taking the simpler form
$x\bigl(\dt_G(S)\bigr)\geq f(S)$ for all $S\sse V$. We denote this LP by $\fndlp$.
Our algorithm is based on iterative rounding, where we iteratively 
fix some $x_e$ variables to $1$. 
Suppose that we 
have fixed 
$x_e=1$ for all $e\in Z_1$.
Then, in that iteration, we are considering the residual problem on the subgraph
$G'=(V,E'=E-Z_1)$, 
which leads to the following variant of \ref{lp}:
\begin{alignat}{3}
\min & \quad & \sum_{e\in E'}c_ex_e \tag{P'} \label{lpp} \\
\text{s.t.} \quad && x\bigl(\dt_{G'}(S)\bigr) & \geq 
\min\bigl\{f(S),|\dt_G(S)|\bigr\}-|\dt_{Z_1}(S)| 
\qquad && \forall S\sse V \label{rescon} \\
&& 0 \leq x_e & \leq 1 && \forall e\in E'. \notag
\end{alignat}
Defining $f'(S)=f(S)-|\dt_{Z_1}(S)|$ for all $S\sse V$, we can rephrase constraints 
\eqref{rescon} as 
$x\bigl(\dt_{G'}(S)\bigr)\geq\min\bigl\{f'(S),|\dt_{G'}(S)|\}$ for all $S\sse V$. 
Since $f'$ is weakly supermodular, symmetric (Lemma~\ref{wsupm-props}~\ref{wsupm-res}),
and normalized, 
\eqref{lpp} is of the same form as \ref{lp}: we have 
\eqref{lpp} $\equiv \lpname[f',G']$. This 
will be convenient for iterative rounding, as the residual problem is of the same form as
the original problem.
\ref{lp} and \eqref{lpp} involve an exponential number of constraints, but using
the ellipsoid method, we can find an extreme-point optimal solution provided we have a
separation oracle for them. Recall that a separation oracle is a procedure that given a
candidate point determines if it is feasible, or returns a violated 
inequality. Observe that a separation oracle for \ref{lp} also yields a separation
oracle for \eqref{lpp}, since \eqref{lpp} is simply \ref{lp} after fixing some variables to
$1$. 

\begin{fact}[Ellipsoid method~\cite{GrotschelLS93}] \label{ellipsoid}
Given a polytime separation oracle for \ref{lp}, one can find an extreme-point optimal
solution to LPs \ref{lp} and \eqref{lpp} in polytime.
\end{fact}

For \crsndp, i.e., when $f=f^\sndp$, a polytime separation oracle for \ref{lp} can be
obtained using min-cut computations. This was shown by~\cite{DinitzKK22} 
for path-relative \kecss, which
is equivalent to cut-relative \kecss (Lemma~\ref{kecss-equiv}); refining their ideas 
yields a separation oracle for \ref{lp} for general \sndp. 

\begin{lemma} \label{crlp-solve}
When $f=f^\sndp$, there is a polytime separation oracle for \ref{lp}. 
Hence, one can obtain extreme-point optimal solutions to \ref{lp} and
\eqref{lpp} in polytime. 
\end{lemma}

\begin{proof} 
The second statement follows from Fact~\ref{ellipsoid}.
Let $(s_i,t_i,r_i)_{i\in[k]}$ be the $k$ tuples defining the \sndp instance, which recall leads
to the cut-requirement function $f^\sndp(S):=\max\,\{r_i: |S\cap\{s_i,t_i\}|=1\}$ for all
$S\sse V$. Say that $S\sse V$ is a small cut if $f^\sndp(S)>|\dt_G(S)|$, and it is a small
cut for $i$ if $|S\cap\{s_i,t_i\}|=1$ and $r_i>|\dt_G(S)|$.
Note that if $e\in\dt_G(S)$ for any small cut $S$, then we must have $x_e=1$ in any
feasible solution. 

Now we can equivalently phrase the constraints of \ref{lp} as follows:
for any edge $a=uv\in E$, we must have $x_a=1$ or, for every $i\in[k]$, the minimum
$s_i$-$t_i$ cut also separating $u$ and $v$ under the $\{x_e\}_{e\in E}$-capacities should
have capacity at least $r_i$.
So to detect if $x$ is a feasible solution, we simply need to consider every edge $a=uv$
with $x_a<1$
and every $i\in[k]$, and find the minimum $s_i$-$t_i$ cut also separating $u$ and $v$
under the $\{x_e\}_{e\in E}$ capacities. 
If an $s_i$-$t_i$ cut separates $u,v$, then $u$ lies on the $s_i$-side and $v$ lies on the
$t_i$-side, or the other way around.
So to find the minimum $s_i$-$t_i$ cut also separating $u$ and $v$, we find the minimum
$\{s_i,u\}$-$\{t_i,v\}$ cut---by which we mean the minimum $s_i$-$t_i$ cut where $s_i,u$
are on the same side, and $t_i,v$ are on the other side--and the minimum
$\{s_i,v\}$-$\{t_i,u\}$-cut and take the one with smaller capacity.  
\end{proof}

\subsubsection*{Extreme-point solution to \boldmath \ref{lp} not defined by a laminar family}
\label{append-examples} \label{append-nolaminar} \label{nolaminar}

We show that for the \crsndp instance shown in Fig.~\ref{nolaminar-fig}, the extreme-point
solution $\hx$ shown in the figure, 
cannot be expressed as the solution to {\em any laminar family of tight cut constraints}
\eqref{cuts} (along with tight edge constraints \eqref{edges}). 
The \sndp instance is $s$-$t$ $2$-edge connectivity: that is, the base
cut-requirement function is given by $f(S)=2$ if $S$ is an $s$-$t$ cut, and $f(S)=0$
otherwise.  

\begin{figure}[ht!]
\centering
\includegraphics[width=2.25in]{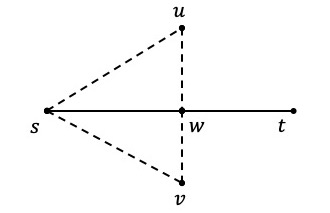}
\caption{\crsndp instance, and an extreme-point solution $\hx$ to \ref{lp}. 
The solid edges have $\hx_e=1$ and the dashed edges have $\hx_e=0.5$, so
$\hx_e=0.5$ for $e\in\{su,uw,sv,vw\}$ and $\hx_e=1$ for $e\in\{sw,wt\}$.
}
\label{nolaminar-fig}
\end{figure}

Note that the only small cut above containing $s$ is $S=\{s,u,v,w\}$. 
We argue that: (1) $\hx$ is an extreme-point solution to \ref{lp}; and (2) $\hx$ cannot be
defined using a laminar family of cut constraints. 

\vspace*{-2ex}
\subparagraph*{\boldmath $\hx$ is an (optimal) extreme point.}
It suffices to show that $\hx$ is feasible, and
that it is the unique solution to a collection of tight constraints. It is straightforward
to verify feasibility of $\hx$ by inspection. Consider the following set of equations,
which come from constraints that are tight at $\hx$.
\begin{alignat}{1}
x\bigl(\dt(S)\bigr) & =2 \qquad \forall S\in\bigl\{\{s\},\{s,u,v\},\{s,u,w\},\{s,v,w\}\bigr\},
\label{tightcuts} \\
x_e & =1 \qquad \forall e\in\{sw,wt\}. \label{tightedges}
\end{alignat}
We claim that $\hx$ is the unique solution to these equations, showing that $\hx$ is an
extreme-point solution to \ref{lp}. Indeed, after substituting $x_{sw}=x_{wt}=1$, by
considering 
every pair of sets in \eqref{tightcuts}, we obtain that all edges in $\{su,uw,sv,vw\}$
must have the same $x_e$ value. (For example, 
$x\bigl(\dt(s)\bigr)=2=x\bigl(\dt(\{s,u,w\})\bigr)$ implies that $x_{vw}=x_{su}$.) 
Together with $x\bigl(\dt(s)\bigr)=2$ and $x_{sw}=1$, this fixes the values of all edges
in $\{su,uw,sv,vw\}$ to $0.5$.

Taking edge costs $c_{sw}=0$ and $c_e=1$ for all other edges, it is not hard to see
that $\hx$ is an optimal solution. The constraints 
$x\bigl(\dt(S)\bigr)\geq 2$ for all
$S\in\bigl\{\{s\},\{s,u,v\},\{s,u,w\},\{s,v,w\}\bigr\}$ imply that (since $x_e\leq 1$
for all $e$)
\[
x_{su}+x_{sv}\geq 1, \quad x\bigl(\dt(u)\bigr)\geq 1, \quad
x\bigl(\dt(v)\bigr)\geq 1, \quad x_{uw}+x_{vw}\geq 1.
\]
Adding these constraints, we obtain that $x_{su}+x_{sv}+x_{uv}+x_{vw}\geq 2$. Therefore,
the cost of any feasible solution is at least $3$, and we have $c^T\hx=3$.

\vspace*{-1ex}
\subparagraph*{\boldmath $\hx$ cannot be defined via a laminar family.}
We now argue that $\hx$ cannot be defined by any laminar family of tight cut constraints
and tight edge constraints. 
Note that any tight cut constraint corresponds to an $s$-$t$
cut. So a laminar family $\Lc$ of tight cut constraints consists of a union of two chains
$\Lc_1\cup\Lc_2$, where a chain is a nested collection of sets: $\Lc_1$ comprising sets
containing $s$ but not $t$, and $\Lc_2$ comprising sets containing $t$ but not $s$. 
Then $\Lc'=\Lc_1\cup\{V-S: S\in\Lc_2\}$ is also a laminar family, yielding the same
cut-constraints as $\Lc$, and and all sets in $\Lc'$contain $s$ but not $t$. 
So we only need to consider laminar families of tight cut constraints corresponding to a  
chain of sets containing $s$ and not $t$.

A simple argument now rules out the existence of any such chain that can
uniquely define $\hx$ along with the tight edge constraints \eqref{tightedges}. 
Any such chain of sets $\C$ is obtained by ordering the elements $u,v,w$, and taking
sets of the form $\{s\}\cup T$, where $T$ is some prefix (including the empty prefix) of
this ordered sequence; so $\C$ contains at most $4$ sets. 
Since $\hx$ has $4$ fractional edges, $E'=\{su,sv,uw,vw\}$, it must be that
$\bigl\{\chi^{\dt_{E'}(S)}\}_{S\in\C}$ contains $4$ linearly independent vectors.
So $\C$ contains exactly $4$ sets, and in particular the set $\{s,u,v,w\}$,
and all vectors in $\bigl\{\chi^{\dt_{E'}(S)}\}_{S\in\C}$ are linearly independent. This
yields a contradiction, since $\chi^{\dt_{E'}(\{s,u,v,w\})}=\vec{0}$.

\section{Modeling power of \crfndp} \label{append-modeling} \label{modeling}

We prove that the \crfndp examples listed under ``Modeling power'' in
Section~\ref{intro} (Examples~\ref{graceful} and~\ref{gengraceful}) model the stated
problems. 
Recall that in both examples, 
we have a non-increasing function $\thresh:\Z_+\mapsto\Z_+$. 

\begin{theorem} \label{graceful-thm}
Consider the base cut-requirement function $f^{\grace}$ in Example~\ref{graceful}, 
where $f^\grace(S):=\min\,\{\ell: \thresh(\ell)<|S|\}$, for $S\neq\es$. 

\begin{enumerate}[label=(\alph*), topsep=0ex, itemsep=0.1ex, leftmargin=*]
\item The function $f^\grace$ is weakly-supermodular. 

\item $H$ is feasible for
$\crfndp[f^{\grace}]$ iff for every fault-set $F\sse E$, and every $S\sse V$ with
$|S|\leq\thresh(|F|)$, $S$ is (the node-set of) a component of $G-F$ iff it is a component
of $H-F$.
\end{enumerate}
\end{theorem}

\begin{proof}
The function $f^{\grace}$ is {\em downwards-monotone}: if $\es\neq T\sse S$, then 
$f^\grace(T)\geq f^{\grace}(S)$. This follows from the definition, since $\thresh$ is a 
non-increasing function.
Weak-supermodularity of $f^{\grace}$ is now immediate since we have
$f^\grace(A)+f^\grace(B)\leq f^\grace(A-B)+f^\grace(B-A)$, if $A-B, B-A\neq\es$, and
otherwise $f^\grace(A)+f^\grace(B)=f^\grace(A\cap B)+f^\grace(A\cup B)$.

Note that $f^\grace$ is not symmetric or normalized, but as indicated by
Lemma~\ref{wsupm-props}, we can consider \crfndp with the symmetric, normalized version of 
$f^{\grace}$ as the base cut-requirement function, without changing the problem.

To prove the feasibility characterization in part (b), we note that the function
$f^{\grace}$ is constructed so that for any $S\sse V$ and $F\sse E$, we have 
$|S|\leq\thresh(|F|)$ iff $f^{\grace}(S)>|F|$, so that $|\dt_H(S)|\geq f^{\grace}(S)$
implies that $S$ is not a component of $H-F$. 
Suppose $H$ is such that
$|\dt_H(S)|\geq\min\bigl\{f^{\grace}(S),|\dt_G(S)|\bigr\}$ for all $S\sse V$. 
Consider any fault-set $F\sse E$ and any $S\sse V$ such that $|S|\leq\thresh(|F|)$.  
By the above, if $S$ is a component of $H-F$, then $|\dt_H(S)|<f^{\grace}(S)$, and
so $S$ must also be a component of $G-F$. But this also implies that if $S$ is a component
of $G-F$, then it is a component of $H-F$; otherwise, there is some 
$A\subsetneq S$ that is a component of $H-F$ with $|A|<|S|\leq\thresh(|F|)$, and so $A$
must be a component of $G-F$, which is not the case.

Conversely, suppose that for every $F\sse E$, and $S\sse V$ with
$|S|\leq\thresh(|F|)$, we have that $S$ is a component of $G-F$ iff it is a component of
$H-F$. 
If $|\dt_H(S)|<\min\bigl\{f^\grace(S),|\dt_G(S)|\bigr\}$ for some $S\sse V$,
then take $F=\dt_H(S)$. Since $f^\grace(S)>|F|$, we have $|S|\leq\thresh(|F|)$. But $S$
is a component of $H-F$ and not $G-F$, which yields a contradiction.
\end{proof}

\begin{theorem} \label{gengraceful-thm}
Let $\prop:2^V\mapsto\R_+$ be a monotone function, i.e., $\prop(T)\leq\prop(S)$ if 
$T\sse S$, as in Example~\ref{gengraceful}.
Consider \crfndp, where the base cut-requirement function $f$ is given by 
$f(S):=\min\,\{\ell: \thresh(\ell)<\prop(S)\}$, for all $S\neq\es$.
Then
\begin{enumerate}[label=(\alph*), topsep=0.1ex, noitemsep, leftmargin=*]
\item $f$ is weakly-supermodular; 
\item $H$ is feasible for $\crfndp$ iff for every fault-set $F\sse E$, and every $S\sse V$
with $\prop(S)\leq\thresh(|F|)$, $S$ is (the node-set of) a component of $G-F$ iff it is a
component of $H-F$.
\end{enumerate}
\end{theorem}

\begin{proof}
We mimic the proof of Theorem~\ref{graceful-thm}. Since $\prop$ is monotone, we 
obtain that $f$ is downwards-monotone, and hence weakly-supermodular.

For part (b), similar to before, we have that $\prop(S)\leq\thresh(|F|)$ iff $f(S)>|F|$. 
Suppose $H$ is feasible for \crfndp.
Consider any fault-set $F\sse E$ and any $S\sse V$ be such that $\prop(S)\leq\thresh(|F|)$.  
By the above, $f(S)>|F|$, so if $S$ is a component of $H-F$, then it must also be a
component of $G-F$. Again, this also implies that if $S$ is a component of $G-F$, then it
is a component of $H-F$: otherwise, there is some $A\subsetneq S$ that is a component of
$H-F$, and $\prop(A)\leq\prop(S)\leq\thresh(|F|)$, which means that $A$ is a component of $G-F$.

Conversely, suppose that for every $F\sse E$, and $S\sse V$ with
$\prop(S)\leq\thresh(|F|)$, we have that $S$ is a component of $G-F$ iff it is a component
of $H-F$. 
If $|\dt_H(S)|<\min\bigl\{f(S),|\dt_G(S)|\bigr\}$ for some $S\sse V$,
then take $F=\dt_H(S)$. Since $f(S)>|F|$, we have $\prop(S)\leq\thresh(|F|)$. But $S$
is a component of $H-F$ and not $G-F$, which yields a contradiction.
\end{proof}

Our intent with Examples~\ref{graceful} and~\ref{gengraceful} is to illustrate that
\crfndp can be used to model some interesting problems. 
The choice of the $\thresh$ and/or $\prop$ functions can of course impact our ability to 
solve the LP-relaxation \ref{crlp} efficiently.

\section{Structure of feasible \boldmath \crfndp solutions: a decomposition result}
\label{decomp}
We now prove our decomposition result, showing that \crfndp can be reduced to \fndp
with suitably-defined weakly-supermodular cut-requirement functions.
Recall that the base cut-requirement function $f$ is weakly supermodular (and symmetric,
normalized), but the cut-requirement function $g^{\crfndp}$ underlying \crfndp need not be.
Given $x\in\R^E$ and $S\sse V$, we use $x^S$ to denote the
restriction of $x$ to edges in $S$, i.e., the vector $(x_e)_{e\in E(S)}\in\R^{E(S)}$.

\begin{theorem}[{\bf Decomposition result}] \label{redn-cor} 
There is a partition $V_1,V_2,\ldots,V_k$ of $V$ and weakly-supermodular, symmetric, and
normalized functions $f_i:2^{V_i}\mapsto\Z$ such that the following hold.

\begin{enumerate}[label=(\alph*), topsep=0.1ex, itemsep=0.2ex, leftmargin=*]
\item $F\sse E$ is a feasible solution to $\crfndp[(f,G)]$ iff
$F\supseteq\dt_G(V_i)$ and $F(V_i)$ is a feasible solution to $\fndp[{(f_i,\gs[V_i])}]$ for
all $i\in[k]$.

\item $x\in\R^E$ is feasible to $\lpname[f,G]$ iff 
for all $i\in[k]$, we have
$x_e=1$ for all $e\in\dt_G(V_i)$, and $x^{V_i}$ is feasible to $\fndlp[{f_i,\gs[V_i]}]$.

\item 
$x$ is an extreme-point solution to $\lpname[f,G]$ iff
for all $i\in[k]$, we have
$x_e=1$ for all $e\in\dt_G(V_i)$, and 
$x^{V_i}$ is an extreme-point solution to $\fndlp[{f_i,\gs[V_i]}]$. 
\end{enumerate}
\end{theorem}

Before delving into the proof of Theorem~\ref{redn-cor}, 
we state the following
immediate corollary of Theorem~\ref{redn-cor} (c), which directly leads to a
$2$-approximation algorithm for \crfndp whenever \ref{lp} can be solved efficiently.  

\begin{corollary} \label{extpoint}
Let $\hx$ be an extreme-point solution to $\lpname[f,G]$. There exists
some $e\in E$ for which $\hx_e\geq 1/2$.
\end{corollary}

\begin{proof}
By Theorem~\ref{redn-cor} (c),
there is a partition of $V$ and \fndp-instances defined on each part such that $\hx$
induces an extreme-point solution in each \fndp-instance. By Jain's result, this
implies that in each of these instances, there is some edge $e$ with $\hx_e\geq 1/2$.
\end{proof} 

The rest of this section is devoted to the proof of Theorem~\ref{redn-cor}.
For a set $S\sse V$, define its 
{\em deficiency} $\defc(S):=f(S)-|\dt_G(S)|$. We use $\defc_{f,G}(S)$ (or $\defc_{f,E}(S)$)
when we want to make the cut-requirement function and graph (or edge-set) explicit. 
By Lemma~\ref{wsupm-props}~\ref{wsupm-res},
$\defc(S)$ is weakly supermodular; also, it is symmetric (and normalized). We say that $S$
is a {\em small cut} (or more explicitly a {\em small-$(f,G)$-cut}) if $\defc(S)>0$. Clearly,
if there is no small cut, then \crfndp becomes equivalent to \fndp. 
Our chief insight and key structural result is that if there is a small cut, then we can
simplify the \crfndp instance and move to smaller \crfndp instances, by ``splitting'' the
instance along a suitable small cut (Theorem~\ref{main-redn}). 
By repeating this operation, 
we eventually obtain an instance with no small cuts, which leads to Theorem~\ref{redn-cor}.
The splitting operation relies on the following modification of $f$.

\begin{definition} \label{restriction} 
Let $\es\neq S\subsetneq V$. Recall that $\bS:=V-S$. 
Let $Z=\dt_G(S)$.
Define the function $f_S:2^S\mapsto\Z$, 
which we call the {\em restriction of $f$ to $S$},
as follows.
\[
f_S(T)=\begin{cases}
0; & \text{if $T=\es$ or $T=S$} \\
\max\bigl\{f(T)-|\dt_Z(T)|,f(T\cup\bS)-|\dt_Z(T\cup\bS)|\bigr\} & \text{otherwise}
\end{cases} \qquad \forall T\sse S. 
\]
\end{definition}

The intuition here is that for $T\sse S$, $f_S(T)$ is a lower bound on the residual
requirement that needs to be met from $\dt_{\gs}(T)$ given that we pick all edges of
$Z$. 
Note that since $f$ and $|\dt_Z(.)|$ are symmetric,
$f(T\cup\bS)-|\dt_Z(T\cup\bS)|=f(S-T)-|\dt_Z(S-T)|$, and 
so $f_S$ is the symmetric, normalized version of of the residual requirement
function $f-|\dt_Z(\cdot)|$. We use the definition above, as this will be more convenient 
to work with later in Section~\ref{decomp-props}, where we establish other useful
properties of the decomposition given  by Theorem~\ref{redn-cor}. 

In the splitting operation, we pick a small cut $\scut$ of {\em maximum} deficiency, and
move to \crfndp instances on the smaller graphs $\gs[\scut]$, and $\gs[\scutc]$
obtained by restricting $f$ to $\scut$, and restricting $f$ to $\scutc$ respectively 
(Theorem~\ref{main-redn}).  
Lemma~\ref{split-wsupm} states some basic properties of the restriction of $f$. Part (b)
will be quite useful in proving that the splitting operation maintains feasibility, and part
(a) allows us to repeat the splitting operation on the smaller instances.
For a set $S\sse V$, and set $F\sse E$ of edges, we use $F(S)$ to denote the edges of $F$
with both ends in $S$. 

\begin{lemma} \label{split-wsupm}
Let $\es\neq S\subsetneq V$. Let $f_S$ be the restriction of $f$ to $S$. 
Let $\dt_G(S)\sse F\sse E$.

\begin{enumerate}[label=(\alph*), topsep=0.1ex, itemsep=0.1ex, leftmargin=*]
\item $f_S$ is weakly supermodular, symmetric, and normalized.

\item Let $\es\neq T\subsetneq S$. 
Then, $\defc_{f_S,F(S)}(T)=\max\bigl\{\defc_{f,F}(T),\defc_{f,F}(T\cup\bS)\bigr\}$.
\end{enumerate}
\end{lemma} 

\begin{proof} 
Let $Z=\dt_G(S)$.
Define $f':2^V\mapsto\R$ by $f'(T):=f(T)-|\dt_Z(T)|$. Then, $f'$ is weakly supermodular
and symmetric (Lemma~\ref{wsupm-prop}~\ref{wsupm-res}). 

For part (a), we can proceed in two ways. As noted earlier, $f_S$ is the symmetric,
normalized version $f'$, so part (a) follows from Lemma~\ref{wsupm-prop}~\ref{wsupm-sym}. 
Alternatively, observe that $f_S$ is the projection of $f'$
onto $S$ with respect to $\scol:=\{\es,\bS\}$, which 
is closed under set intersections, unions, differences, and complements with
respect to $V-S$. So by Lemma~\ref{wsupm-prop}~\ref{wsupm-proj}, we obtain that
$f_S$ is weakly supermodular and symmetric. Also, $f_S$ is normalized by definition. 

For part (b), the stated equality follows by simply plugging in $f_S$, and noting that for
$Y\in\{T,T\cup\bS\}$, we have $|\dt_F(Y)|=|\dt_{F(S)}(T)|+|\dt_Z(Y)|$. 
\end{proof}

\begin{theorem}[{\bf Splitting along a small cut}] \label{main-redn}
Suppose that there is some small cut. Let $\scut\sse V$ be 
a maximum-deficiency cut. 
Note that $\defc(\scut)>0$ and $\es\neq\scut\subsetneq V$. 
\begin{enumerate}[label=(\alph*), topsep=0.2ex, itemsep=0.2ex, leftmargin=*]
\item 
$F\sse E$ is a feasible solution to $\crfndp[(f,G)]$ iff
$F\supseteq\dt_G(\scut)$, $F(\scut)$ is a feasible solution to $\crfndp[{(f_{\scut},\gs[\scut])}]$,
and $F(\scutc)$ is a feasible solution to $\crfndp[{(f_{\scutc},\gs[\scutc])}]$.

\item 
$x\in\R_+^E$ is feasible to $\lpname[f,G]$ iff $x_e=1$ for all $e\in\dt_G(\scut)$, 
$x^{\scut}$ is feasible to $\lpname[{f_{\scut},\gs[\scut]}]$, and 
$x^{\scutc}$ is feasible to $\lpname[{f_{\scutc},\gs[\scutc]}]$.
\end{enumerate}
\end{theorem}

\begin{proof}
Part (a) can be seen as the special case of part (b), where $x$ is integral.
We prove part (a) here, as the proof is somewhat (notational) simpler. The proof of part
(b) is quite similar, and is deferred to 
Appendix~\ref{append-decomp}. For notational simplicity, let 
$(f_1,G_1,F_1)=\bigl(f_{\scut},\gs[\scut],F(\scut)\bigr)$ and
$(f_2,G_2,F_2)=\bigl(f_{\scutc},\gs[\scutc],F(\scutc)\bigr)$.
Let $Z=\dt_G(\scut)$.

The ``only if'' direction follows easily from the definition of the restriction of $f$.
Let $F$ be feasible to $\crfndp[(f,G)]$. 
Since $\scut$ is a small cut, we must have 
$F\supseteq Z$. 
Consider any $\es\neq T\subsetneq\scut$. Since $F$ is feasible,
we have $|\dt_F(T)|\geq\min\bigl\{f(T),|\dt_G(T)|\bigr\}$, which implies that 
$|\dt_{F_1}(T)|\geq\min\bigl\{f(T)-|\dt_Z(T)|,|\dt_{G_1}(T)|\bigr\}$. Similarly, we have 
$|\dt_F(T\cup\scutc)|\geq\min\bigl\{f(T\cup\scutc),|\dt_G(T\cup\scutc)|\bigr\}$, so
$|\dt_{F_1}(T)|\geq\min\bigl\{f(T\cup\scutc)-|\dt_Z(T\cup\scutc)|,|\dt_{G_1}(T)|\bigr\}$. Combining
the inequalities, 
and using the definition of restriction, we
obtain that $|\dt_{F_1}(T)|\geq\min\bigl\{f_1(T),|\dt_{G_1}(T)|\bigr\}$.
Since $T$ was arbitrary, this shows that $F_1$ is feasible to
$\crfndp[{(f_1,G_1)}]$.

A symmetric argument 
shows that $F_2$ is a feasible solution to $\crfndp[{(f_2,G_2)}]$.

\medskip
Conversely, suppose that $F\supseteq\dt_G(\scut)$, $F_1=F(\scut)$ is a
feasible solution to 
$\crfndp[{(f_1,G_1)}]$, and $F_2=F(\scutc)$ is a feasible solution to
$\crfndp[{(f_2,G_2)}]$. Recall that $Z=\dt_G(\scut)$.
Consider any $\es\neq T\subsetneq V$. 
Lemma~\ref{split-wsupm} (b) can be used to readily argue 
that if $T\sse\scut$ or $T\sse\scutc$, then the constraint for set $T$ in
$\crfndp[(f,G)]$ is satisfied. We show this for $T\sse\scut$; a symmetric argument applies
for $T\sse\scutc$.  
If $T$ is a small-$(f_1,G_1)$-cut, then $F_1\supseteq\dt_{G_1}(T)$. Since
$F\supseteq Z$, this implies that $F\supseteq\dt_G(T)$. 
Otherwise, we have $\defc_{f_1,F_1}(T)\leq 0$, so
by Lemma~\ref{split-wsupm} (b), we have that
$\defc_{f,F}(T)\leq\defc_{f_1,F_1}(T)\leq 0$.
Thus, we always have $|\dt_F(T)|\geq\min\bigl\{f(T),|\dt_G(T)|\bigr\}$.

Now consider $T$ such that $T\cap\scut, T\cap\scutc\neq\es$. 
We exploit the weak supermodularity of $\defc_{f,F}$, and that $\scut$ has
maximum deficiency, to bound $\defc_{f,F}(T)$ in terms of the deficiency of sets that do
not cross $A$, and thus show that the constraint for $T$ is satisfied.
We have

\vspace*{-3ex}
\begin{equation*}
\begin{split}
\defc_{f,F}(T)+\defc_{f,F}(\scut) & \leq\defc_{f,F}(T\cap\scut)+\defc_{f,F}(T\cup\scut)
=\defc_{f,F}(T\cap\scut)+\defc_{f,F}(\scutc-T) \ \ \text{or} \\
\defc_{f,F}(T)+\defc_{f,F}(\scut) & \leq\defc_{f,F}(\scut-T)+\defc_{f,F}(T-\scut)
\end{split}
\vspace*{-2ex}
\end{equation*}

\noindent where in the first inequality we also use the fact that $\defc_{f,F}$ is
symmetric. In both cases, we have 
$\defc_{f,F}(T)+\defc_{f,F}(\scut)\leq\defc_{f,F}(X)+\defc_{f,F}(Y)$
{where $X\sse\scut$, $Y\sse\scutc$ and
$\dt_G(T)\sse\dt_{G}(X)\cup\dt_{G}(Y)\cup Z$}.
As shown previously, for $S\in\{X,Y\}$,  
we have
$|\dt_F(S)|\geq\min\bigl\{f(S),|\dt_G(S)|\bigr\}$, or equivalently, 
$\defc_{f,F}(S)\leq\max\bigl\{0,\defc_{f,G}(S)\bigr\}$.
Note that $\defc_{f,F}(\scut)=\defc_{f,G}(\scut)>0$ since $F\supseteq\dt_G(\scut)$.
Since $\scut$ is a maximum-deficiency small cut, this also implies that
$\defc_{f,F}(\scut)\geq\defc_{f,F}(S)$ for $S\in\{X,Y\}$: we have
$\defc_{f,F}(S)\leq\max\bigl\{0,\defc_{f,G}(S)\bigr\}
\leq\max\bigl\{0,\defc_{f,G}(\scut)\bigr\}=\defc_{f,F}(\scut)$.

If both $X$ and $Y$ are small-$(f,G)$-cuts, then we are done since then we have that
$F\supseteq\dt_{G}(X)\cup\dt_{G}(Y)\cup Z\supseteq\dt_G(T)$.
Otherwise, since $\defc_{f,F}(\scut)\geq\max\bigl\{\defc_{f,F}(X),\defc_{f,F}(Y)\bigr\}$, we
have $\defc_{f,F}(T)\leq\min\bigl\{\defc_{f,F}(X),\defc_{f,F}(Y)\bigr\}\leq 0$, since at
least one of $X,Y$ is not a small-$(f,G)$-cut.
\end{proof}

\begin{proofof}{Theorem~\ref{redn-cor}}
Parts (a) and (b) follow easily by induction on the number of nodes, using parts (a) and (b) of
Theorem~\ref{main-redn}. 
Consider part (a). 
The base case is when there is
no small-$(f,G)$-cut, in which case, we have $V_1=V$, and $f_1=f$. Otherwise, let
$\scut=\argmax_{S\sse V}\defc_{f,G}(S)$. Since $f_{\scut}$ and $f_{\scutc}$ are
weakly-supermodular, we can recurse and induct on $\crfndp[{(f_{\scut},\gs[\scut])}]$ 
and $\crfndp[{(f_{\scutc},\gs[\scutc])}]$ instances, both of which have fewer nodes, to
obtain partitions 
of $\scut$ and 
$\scutc$; we combine these to obtain $V_1,\ldots,V_k$.
Theorem~\ref{main-redn} (a) and
the induction hypothesis then readily yield part (a) here. 
The only thing to note is that 
since the induction hypothesis yields $F\supseteq\dt_{\gs[\scut]}(V_i)$ for a part 
$V_i\sse\scut$, and $F\supseteq\dt_G(\scut)$ by Theorem~\ref{main-redn} (a),
we have $F\supseteq \dt_G(V_i)$. Similarly, $F\supseteq\dt_G(V_i)$ for any part
$V_i\sse\scutc$. 

Part (b) follows similarly from Theorem~\ref{main-redn} (b). 

\medskip
Part (c) follows from part (b), because it is not hard to see, using part (b), that  
any convex combination of feasible solutions to \ref{lp} yields a convex combination of
feasible $\fndlp[{f_i,\gs[V_i]}]$-solutions, and vice versa.

Suppose $x$ is an extreme-point
solution to $\lpname[f,G]$. By feasibility of $x$ and part (b), 
we have $x_e=1$ for all $e\in\dt_G(V_j)$, and $x^{V_j}$ is feasible to
$\fndlp[{f_j,\gs[V_j]}]$ for all $j\in[k]$. 
Suppose that $x^{V_i}$ is not an
extreme-point solution to $\fndlp[{f_i,\gs[V_i]}]$ for some $i\in[k]$. 
Then, we can write $x^{V_i}$ as a convex
combination $\ld y^1+(1-\ld)y^2$ of two solutions $y^1\neq x^{V_i}$, $y^2\neq x^{V_i}$
that are feasible to $\fndlp[{f_i,\gs[V_i]}]$, where $\ld\in[0,1]$. But then 
we can write $x$ as a convex combination of two feasible solutions to $\lpname[f,G]$, as
follows. 
Define $\bx,\hx\in\R^E$, where $\bx_e=\hx_e=x_e$ for all 
$e\in E-E(V_i)$, and $\bx_e=y^1_e$,\, $\hx_e=y^2_e$ for all $e\in E(V_i)$. 
So we have $\bx_e=\hx_e=1$ for all $e\in\dt_G(V_j)$ for all $j\in[k]$, 
$\bx^{V_j}=\hx^{V_j}$ is feasible to $\fndlp[{f_j,\gs[V_j]}]$ for all $j\in[k], j\neq i$;
and $\bx^{V_i}=y^1$, $\hx^{V_i}=y^2$ are feasible to $\fndlp[{f_i,\gs[V_i]}]$.
Again, by part (b), this yields that $\bx$ and $\hx$ are feasible to $\lpname[f,G]$. But
then we have $x=\ld \bx+(1-\ld)\hx$ for feasible solutions $\bx\neq x$, $\hx\neq x$, which
contradicts that $x$ is an extreme-point solution to $\lpname[f,G]$.

Conversely, suppose that $x_e=1$ for all $e\in\dt_G(V_j)$, and $x^{V_j}$ is an extreme-point
solution to $\fndlp[{f_j,\gs[V_j]}]$, for all $j\in[k]$. 
Then, $x$ is feasible to $\lpname[f,G]$ by part (b). 
Suppose $x$ is not an extreme-point solution to $\lpname[f,G]$, and $x=\ld\bx+(1-\ld)\hx$
for feasible solutions $\bx\neq x$, $\hx\neq x$. Then, again by part (b), we have
$\bx_e=\hx_e=1$ for all $e\in\dt_G(V_j)$ for all $j\in[k]$, and so
$\bx^{V_i}\neq\hx^{V_i}$ for some $i\in[k]$. Also, $\bx^{V_i}$ and $\hx^{V_i}$ are
feasible to $\fndlp[{f_i,\gs[V_i]}]$ and $x^{V_i}=\ld\bx^{V_i}+(1-\ld)\hx^{V_i}$,
contradicting that $x^{V_i}$ is an extreme-point solution to $\fndlp[{f_i,\gs[V_i]}]$.
\end{proofof}

One may wonder if the $\fndp[{(f_i,\gs[V_i])}]$-instances in the decomposition of 
Theorem~\ref{redn-cor} can be specified more directly or succinctly.
We prove some properties of the decomposition given by its proof 
in Section~\ref{decomp-props}, 
which as a {\em consequence}, 
enable one to do so (see Lemma~\ref{decomp-partn}). 
But 
we need our splitting approach to argue why this succinct description fulfills the
properties in Theorem~\ref{redn-cor}. 
We discuss this, and other subtleties in Section~\ref{succinct}.
The properties we establish also imply that the decomposition in Theorem~\ref{redn-cor} 
can be computed efficiently 
given suitable algorithmic access to the base cut-requirement function $f$.

\section{Algorithmic implication: LP-relative \boldmath $2$-approximation algorithm} 
\label{algresults}
We now exploit the decomposition result stated in Theorem~\ref{redn-cor}, and
Corollary~\ref{extpoint}, to obtain a $2$-approximation algorithm. 
We assume that we have a separation oracle for \ref{lp}, which implies that we can
efficiently find an extreme-point optimal solution to \ref{lp} and any LP that arises in
a subsequent iteration after fixing some $x_e$ variables to $1$ 
(Fact~\ref{ellipsoid}). 

\bigskip
{\hrule \nopagebreak
\small
\smallskip\noindent 
{\bf Algorithm \boldmath \crndpalg.} The input is 
$\bigl(G=(V,E),\{c_e\geq 0\},f:2^V\mapsto\Z\bigr)$, where $f$ is weakly-supermodular (and
symmetric, normalized). We assume there is a separation oracle for \ref{lp}. 

\begin{enumerate}[label=A\arabic*., topsep=0.2ex, itemsep=0.1ex, leftmargin=*]
\item Initialize $F\assign\es$, $G'=(V,E')\assign G$, and $f'\assign f$.
\item While $F$ is not a feasible solution (which can be detected using the separation
  oracle), repeat:
\begin{enumerate}[label*=\arabic*, nosep, leftmargin=*]
\item Find an extreme-point optimal solution $\hx$ to $\lpname[f',G']$.
\item Let $Z_1=\{e\in E': \hx_e\geq 1/2\}$. 
Update $F\assign F\cup Z_1$, $E'\assign E'-Z_1$, 
and $f'(S)=f'(S)-|\dt_{Z_1}(S)|$ for all $S\sse V$.
\end{enumerate}
\item Return $F$.
\end{enumerate}
\hrule}

\medskip
\begin{remark} \label{alg-remk}
The collection of small cuts {\em does not change across iterations}; that is,
$S$ is a small-$(f,G)$-cut iff it is a small-$(f',G')$-cut in some other iteration. This
is because, for any $S\sse V$, we change $|\dt_{G'}(S)|$ and $f'(S)$ by the
same amount 
in every iteration. 
In the first iteration, 
we have $\hx_e=1$ for every $e\in\dt_G(S)$ and every small cut $S$, 
and so 
$\dt_{G'}(S)=\es$ in any subsequent iteration. 
So we do not have any constraints for small cuts in subsequent iterations, 
and this causes problems with uncrossing tight constraints. Nevertheless, we are able to
argue, due to our decomposition result, that there is always some edge with 
$\hx_e\geq 1/2$.
\end{remark}

\begin{theorem} \label{approx}
Let $\OPT$ denote the optimal value of \ref{lp}. Algorithm \crndpalg returns a feasible
\crfndp-solution $F\sse E$ such that $c(F)\leq 2\cdot\OPT$.
\end{theorem}

\begin{proof}
By Corollary~\ref{extpoint}, at least one edge gets added to $F$ in every iteration, so
the algorithm terminates in at most $|E|$ iterations. The cost bound follows by induction
on  
the number of iterations to termination. The base case is trivial.
Suppose $F=Z_1\cup F'$, where $F'$ is the solution found in subsequent iterations for the
instance $\crfndp[(f',G')]$. Let $\OPT'$ be the optimal solution to
$\lpname[G',f']$. Then, by induction, we have 
$c(F_1)\leq 2\cdot\OPT'$ and $\OPT'\leq\sum_{e\in E'}c_e\hx_e$, since 
$(\hx_e)_{e\in E'}$ is a feasible solution to $\lpname[f',G']$. We also have 
$c(Z_1)\leq 2\cdot\sum_{e\in Z_1}c_e\hx_e$. 
Therefore, $c(F)\leq 2c^T\hx=2\cdot\OPT$.
\end{proof}

\noindent
Since $\lpname[{f^{\sndp},G}]$ admits a polytime separation oracle
(Lemma~\ref{crlp-solve}), we obtain the following result as a corollary. 

\begin{theorem} \label{crsndp-thm}
There is a $2$-approximation algorithm for \crsndp.
\end{theorem}

\section{Properties of the decomposition in Theorem~\ref{redn-cor}}
\label{decomp-comp} \label{append-decompcomp} \label{decomp-props}
Throughout this section, $f$ and $G$ always denote the base
cut-requirement function $f$ and graph $G$ {\em that are given as input to \crfndp}.
Theorem~\ref{redn-cor} does not pinpoint a specific decomposition, and only mentions
some properties that the $(f_i,V_i)$ tuples should satisfy. 
By ``the decomposition in Theorem~\ref{redn-cor},'' we mean the decomposition given by its
proof, wherein we obtain a maximum deficiency set $\scut\sse V$ as in
Theorem~\ref{main-redn}, and recurse on the $\crfndp[{(f_{\scut},\gs[\scut])}]$ and
$\crfndp[{(f_{\scutc},\gs[\scutc])}]$-instances.
This process can be conveniently described in terms of a decomposition
tree as follows. Each node of the tree is labeled by a tuple $(\cdfunc,S)$, 
where $S\sse V$ and $\cdfunc:2^S\mapsto\Z$ is a weakly-supermodular, symmetric, normalized
cut-requirement function on $S$; we will often simply say node $(\cdfunc,S)$. 
The root of the tree is labeled $(f,V)$. For each (current) leaf node $(\cdfunc,S)$, 
we find $\scut\sse S$ such that 
$\defc_{\cdfunc,\gs}(\scut)=\max_{T\sse S}\defc_{\cdfunc,\gs}(T)$. If $\defc_{f,\gs}(\scut)\leq 0$, then there
is no small-$(\cdfunc,\gs)$-cut and $(\cdfunc,S)$ gets fixed as a leaf node. 
Otherwise, we create two children of
$(\cdfunc,S)$ labeled $(\cdfunc_1,\scut)$ and $(\cdfunc_2,S-\scut)$, 
where $\cdfunc_1$ and $\cdfunc_2$ are the
restrictions of $\cdfunc$ to $\scut$ and $S-\scut$ respectively (see
Definition~\ref{restriction}), and continue.
We call this tree the {\em canonical} decomposition tree, and refer to
the decomposition given by its leaves as the canonical decomposition.

We establish some useful properties of the canonical decomposition tree that may be of
independent interest. 
In particular, one can infer from these properties that the canonical decomposition tree
can be computed efficiently given suitable algorithmic access to $f$. 
This suggests another way of 
obtaining an algorithmic result for \crfndp, namely, by computing the canonical
decomposition efficiently and then (approximately) solving the 
$\fndp[{(f_i,\gs[V_i])}]$-instances. 
In particular, if the resulting $f_i$ functions have some additional structure, then 
these $\fndp[{(f_i,\gs[V_i])}]$-instances might admit other, possibly
more-efficient approximation algorithms (e.g., via the primal-dual method), which would
yield an alternate (to iterative rounding) way of obtaining an approximation result for
\crfndp.  

Another consequence of the properties we establish is that they allow us to succinctly
describe the $(f_i,V_i)$ tuples given by the leaves of the canonical decomposition
tree (see Lemmas~\ref{decomp-partn} and~\ref{decomp-fns}), as also a closely-related  
decomposition satisfying the properties stated in Theorem~\ref{redn-cor}. 
While this yields a self-contained, succinct description of the canonical decomposition, 
we emphasize that in order to prove that the $\fndp$-instances given by this succinct
description satisfy the properties stated in Theorem~\ref{redn-cor},
we need to go via the canonical decomposition and exploit Lemma~\ref{decomptree} (and
thus our approach of iteratively simplifying the instance via splitting along
maximum-deficiency small cuts). So having this succinct description does not
simplify the underlying arguments; 
see also Section~\ref{altdecomp-remk}.

\begin{lemma} \label{decomptree}
Let $(\cdfunc,S)$ be a node of the canonical decomposition tree. 
Let the nodes on the path from $(\cdfunc,S)$ to the root be labeled 
$(\cdfunc_0,S_0)=(\cdfunc,S),(\cdfunc_1,S_1),\ldots,(\cdfunc_\ell,S_\ell)=(f,V)$. 
Define $R_I:=\bigcup_{i\in I}(S_i-S_{i-1})$. Note that $R_{\es}:=\es$. 
For $i\in[\ell]$, we abbreviate $R_{\{i\}}$ to $R_i$, so $R_i=S_i-S_{i-1}$.
\begin{enumerate}[label=(\alph*), topsep=0.1ex, itemsep=0.1ex, leftmargin=*]
\item For any $\es\neq T\subsetneq S$, we have 
$\defc_{\cdfunc,\gs}(T)=\max\,\bigl\{\defc_{f,G}(T\cup R_I): I\sse[\ell]\bigr\}$.

\item For any $i,j\in\{0\}\cup[\ell]$ with $i\leq j$, and any $T\sse S_i$, we have
$\defc_{\cdfunc_j,\gs[S_j]}(T)\leq\defc_{\cdfunc_i,\gs[S_i]}(T)$.

\item Let $s,t\in S$. 
We have
\vspace*{-1ex}
\[
\max\,\bigl\{\defc_{\cdfunc,\gs}(T):\ T\sse S\text{ is an $s$-$t$ cut}\bigr\}
= \max\,\bigl\{\defc_{f,G}(T):\ T\sse V\text{ is an $s$-$t$ cut}\bigr\}.
\]
\noindent
Moreover, if an $s$-$t$ cut $T\sse V$ attains the maximum on the right, then $T\cap S$ 
attains the maximum on the left.
\end{enumerate}
\end{lemma}

\begin{proof}
We prove part (a) by induction on $\ell$, which we call the depth of $(\cdfunc,S)$. 
The base case when $\ell=0$ 
is trivially true. Suppose the statement holds whenever the depth is at most $d$,
and consider a node $(\cdfunc,S)$ with depth $\ell=d+1$. 
Then $\cdfunc$ is the restriction of $\cdfunc_1$ to $S$, so letting $Z=\dt_{\gs[S_1]}(S)$, we have 
\begin{alignat}{1}
\defc_{\cdfunc,\gs}(T) & =
\max\bigl\{\cdfunc_1(T)-|\dt_{Z}(T)|,\cdfunc_1(T\cup S_1-S)-|\dt_Z(T\cup S_1-S)|\bigr\}-|\dt_{\gs}(T)|
\notag \\
& = \max\bigl\{\defc_{\cdfunc_1,\gs[S_1]}(T),\defc_{\cdfunc_1,\gs[S_1]}(T\cup R_1)\bigr\}.
\label{decompineq1}
\end{alignat}
By the induction hypothesis, we have 
\begin{equation*}
\begin{split}
\defc_{\cdfunc_1,\gs[S_1]}(T) & =\max\,\bigl\{\defc_{f,G}(T\cup R_I): I\sse\{2,\ldots,\ell\}\bigr\},
\qquad \text{and} \\
\defc_{\cdfunc_1,\gs[S_1]}(T\cup R_1) & =\max\,\bigl\{\defc_{f,G}(T\cup R_1\cup R_I): I\sse\{2,\ldots,\ell\}\bigr\}.
\end{split}
\end{equation*}
Plugging these in \eqref{decompineq1}, and noting that $R_1\cup R_I=R_{\{1\}\cup I}$ for
$I\sse\{2,\ldots,\ell\}$ yields 
$\defc_{\cdfunc,\gs}(T)=\max\,\bigl\{\defc_{f,G}(T\cup R_I): I\sse[\ell]\bigr\}$,
completing the induction step, and the proof.

\medskip
Part (b) follows directly from part (a), since the collection of sets 
$\bigl\{T\cup R_I: I\sse\{i+1,\ldots,\ell\}\bigr\}$ considered in the computation of
$\defc_{\cdfunc_i,\gs[S_i]}(T)$ is a superset of the collection 
$\bigl\{T\cup R_I: I\sse\{j+1,\ldots,\ell\}\bigr\}$ considered for
$\defc_{\cdfunc_j,\gs[S_j]}(T)$. 

\medskip
We now prove part (c).
Let LHS = $\max\,\bigl\{\defc_{\cdfunc,\gs}(T):\ T\sse S\text{ is an $s$-$t$ cut}\bigr\}$
and RHS = $\max\,\bigl\{\defc_{f,G}(T):\ T\sse V\text{ is an $s$-$t$ cut}\bigr\}$.
By part (a), for any $s$-$t$ cut $T\sse S$, we can find some $s$-$t$ cut $T_1\sse V$
such that $\defc_{\cdfunc,\gs}(T)=\defc_{f,G}(T_1)$. So LHS $\leq$ RHS.

For the other direction, 
fix any $s$-$t$ cut $T_1\sse V$. 
We show that $\defc_{f,G}(T_1)\leq\defc_{\cdfunc,\gs}(T_1\cap S)$, 
which shows that RHS $\leq$ LHS, and also yields the desired statement 
relating maximizers of the RHS and the LHS.

Recall that $(\cdfunc_1,S_1),(\cdfunc_2,S_2),\ldots,(\cdfunc_\ell,S_\ell)=(f,V)$ are the nodes on the path
from $(\cdfunc,S)$ to the root, with $(\cdfunc_1,S_1)$ being the parent of $(\cdfunc,S)$.
Let $a\in\{0,1,\ldots,\ell\}$ be the smallest index such that $T_1\sse S_a$.
Define $X_i:=T_1\cap S_i$ for all $i\in[a]$ (so $X_a:=T_1$).
We exploit the weak-supermodularity (and symmetry) of the $\defc_{\cdfunc_i,\gs[S_i]}(\cdot)$ 
functions to show that 
\begin{equation}
\defc_{\cdfunc_i,\gs[S_i]}(X_i)\leq\defc_{\cdfunc_{i-1},\gs[S_{i-1}]}(X_{i-1}) \qquad 
\forall i\in[a].
\label{decomp-inv}
\end{equation}
Clearly, \eqref{decomp-inv} implies that 
$\defc_{\cdfunc_a,\gs[S_a]}(T_1)\leq\defc_{\cdfunc,\gs}(T_1\cap S)$.
We also have $\defc_{f,G}(T_1)\leq\defc_{\cdfunc_a,\gs[S_a]}(T_1)$, which finishes the proof.

We now prove \eqref{decomp-inv}.
We have $\defc_{\cdfunc_i,\gs[S_i]}(X_i)+\defc_{\cdfunc_i,\gs[S_i]}(S_{i-1})\leq
\defc_{\cdfunc_i,\gs[S_i]}(U_{i-1})+\defc_{\cdfunc_i,\gs[S_i]}(Y_{i-1})$
where $U_{i-1}=X_i\cap S_{i-1}$, $Y_{i-1}=X_i\cup S_{i-1}$, or $U_{i-1}=S_{i-1}-X_i$, 
$Y_{i-1}=X_i-S_{i-1}$.
By the definition of $S_{i-1}$, we have 
$\defc_{\cdfunc_i,\gs[S_i]}(S_{i-1})\geq\defc_{\cdfunc_i,\gs[S_i]}(Y)$ for all 
$Y\sse S_i$. So the above inequality yields that 
$\defc_{\cdfunc_i,\gs[S_i]}(X_i)\leq\defc_{\cdfunc_i,\gs[S_i]}(U_{i-1})$.
Also, by part (b), we have 
$\defc_{\cdfunc_i,\gs[S_i]}(U_{i-1})\leq\defc_{\cdfunc_{i-1},\gs[S_{i-1}]}(U_{i-1})$.
We know that $U_{i-1}$ is $X_i\cap S_{i-1}=X_{i-1}$ or $S_{i-1}-X_i$. 
So since $\defc_{\cdfunc_{i-1},\gs[S_{i-1}]}(\cdot)$ is symmetric, we therefore obtain that
$\defc_{\cdfunc_{i-1},\gs[S_{i-1}]}(U_{i-1})=\defc_{\cdfunc_{i-1},\gs[S_{i-1}]}(X_i\cap S_{i-1})
=\defc_{\cdfunc_{i-1},\gs[S_{i-1}]}(X_{i-1})$. 
Combining everything, we obtain that \eqref{decomp-inv} holds for index $i$.  
\end{proof}

Lemma~\ref{decomptree} (c) indicates that finding maximum-deficiency $s$-$t$
cuts is the natural algorithmic problem that one needs to solve in order to find the
canonical decomposition tree, and the resulting decomposition.
Formally, in the {\em max-deficiency $s$-$t$ cut} problem, we are
given $s,t\in V$ and we need to return an $s$-$t$ cut $\scut\sse V$ that maximizes
$\defc_{f,G}(T)$ over all $s$-$t$ cuts $T\sse V$ 
(and its deficiency $\defc_{f,G}(\scut)$).
More generally, in the {\em weighted max-deficiency $s$-$t$ cut problem}, we are also given
$w\in\R_+^E$, and we want find an $s$-$t$ cut $\scut\sse V$ that
maximizes $\wdef_{f,G}(S):=f(T)-w\bigl(\dt_{G}(T)\bigr)$ over all $s$-$t$ cuts $T\sse V$.

\begin{theorem} \label{poly-decomp}
Given a polytime algorithm $\alg$ for the max-deficiency $s$-$t$ cut problem,
we can efficiently compute the canonical decomposition tree, 
and hence the decomposition in Theorem~\ref{redn-cor}.
\end{theorem}

\begin{proof}
We only need to show that given a node $(\cdfunc,S)$ of the current canonical decomposition tree, 
how to find $\scut\sse S$ such that $\defc_{\cdfunc,\gs}(\scut)=\max_{T\sse S}\defc_{\cdfunc,\gs}(T)$. 
Fix some $s\in S$.
Clearly, if we can find $\max\,\{\defc_{\cdfunc,\gs}(T): T\sse S\text{ is an $s$-$t$ cut}\}$,
then we can find $\scut$ by taking the maximum over these $s$-$t$ cuts. 
Using Lemma~\ref{decomptree} (c), we can optimize over $s$-$t$ cuts contained in $S$ by
finding a max $s$-$t$ deficient cut with respect to $(f,G)$ using $\alg$, and taking its
intersection with $S$. 
\end{proof}

Furthermore, note that if can solve weighted max $s$-$t$ deficient cut in polytime, then
for any leaf $(f_i,S)$ of the canonical decomposition tree, we can give an efficient
separation oracle for the LP $\fndlp[{f_i,\gs}]$. This is because exactly as in
Lemma~\ref{decomptree} (c), we have
$\max\,\bigl\{\wdef_{f_i,\gs}(T):\ T\sse S\text{ is an $s$-$t$ cut}\bigr\}
= \max\,\bigl\{\wdef_{f,G}(T):\ T\sse V\text{ is an $s$-$t$ cut}\bigr\}$, and if $T$
attains the maximum on the right, then $T\cap S$ attains the maximum on the left. 
So we can find $\max_{T\sse S}\bigl(f_i(T)-x\bigl(\dt_{\gs}(T)\bigr)$ by performing $|S|$ 
weighted max-$s$-$t$-deficient-cut computations (with respect to $(f,G)$).

For \crsndp, where the base cut-requirement function is $f^\sndp$, finding a weighted
$s$-$t$ deficient cut amounts to a suitable min-cut problem. Recall that the \sndp input
specifies $k$ tuples $(s_i,t_i,r_i)$, where $s_i,t_i\in V$, and $r_i\in\Z_+$ for all $i$,
and $f^\sndp(S)=\max\,\{r_i: |S\cap\{s_i,t_i\}|=1\}$. 
So to solve $\max\,\{f^{\sndp}(T)-w\bigl(\dt_G(T)\bigr):\ T\sse V\text{ is an $s$-$t$ cut}\bigr\}$,
we consider each $i\in[k]$, and 
find $T_i$ that minimizes $w\bigl(\dt_G(T)\bigr)$ over all sets $T\sse V$ such that
$|T\cap\{s_i,t_i\}|=1=|T\cap\{s,t\}|$, by solving a min-cut problem.
We then return $\max_{i\in[k]}\bigl(r_i-w\bigl(\dt_G(T_i)\bigr)$. We therefore obtain the
following corollary.

\begin{corollary} \label{decomp-sndp}
For \crsndp, we can find the canonical decomposition tree in polytime. 
Also, for any leaf $(f_i,S)$ of this tree, we have a polytime separation oracle for
$\fndlp[{f_i,\gs}]$. 
\end{corollary}

\subsection{Succinct description of the canonical decomposition} \label{succinct}
As noted earlier, one structural consequence of Lemma~\ref{decomptree} is that one can
give a succinct description of the $(f_i,V_i)$ tuples obtained at the leaves of the
canonical decomposition tree. This also leads to an alternate decomposition satisfying the
properties stated in Theorem~\ref{redn-cor}. 
Define $\smcuts:=\bigcup_{S\sse V:\defc_{f,G}(S)>0}\dt_G(S)$ to be the set of all 
edges on the boundary of some small-$(f,G)$-cut; we sometimes refer to $\smcuts$ as small-cut
edges. Let $\spart_1,\ldots,\spart_p$ be the node-sets of the components of $G-\smcuts$. 
Note that for a component $\spart_i$, while $\dt_G(\spart_i)\sse\smcuts$, 
{\em it need not be that $\spart_i$ is a small-$(f,G)$-cut}. (For instance, if the base
cut-requirement function $f$ models an $s$-$t$ connectivity problem, then every
small-$(f,G)$-cut must be an $s$-$t$ cut; so unless there are only two small cuts,
it will not be that every $\spart_i$ set is a small cut.) 

Say that components $\spart_i$ and $\spart_j$ are {\em non-separable} if for every
small-$(f,G)$-cut $\scut$, we have that $\spart_i\sse\scut$ iff
$\spart_j\sse\scut$. Clearly, non-separability is an equivalence relation. We show
that the $V_1,\ldots,V_k$ partition is precisely the partition formed by the equivalence
classes of the non-separability relation.
We can also succinctly specify the $f_i:2^{V_i}\mapsto\Z$ functions for $i\in[k]$
(Lemma~\ref{decomp-fns}).  

\begin{lemma} \label{decomp-partn}
The leaf sets $V_1,\ldots,V_k$ of the canonical decomposition are the
equivalence classes of the non-separability relation. That is, for each $V_i$, we have a
set $\eqclass_i\sse[p]$ of indices such that: (i) $V_i=\bigcup_{j\in\eqclass_i}\spart_j$; and
(ii) for any two indices $j\in\eqclass_i$, $j'\in[p]$, components $\spart_j$ and
$\spart_{j'}$ are non-separable iff $j'\in\eqclass_i$.
In particular, note that there are no edges of $\smcuts$ internal to a leaf-set.
\end{lemma}

\begin{proof}
We first argue that the 
$\spart_1,\ldots,\spart_p$ partition is a refinement of the $V_1,\ldots,V_k$ partition.
To do so, we need to show that $\dt_G(V_i)\sse\smcuts$ for every leaf-set $V_i$. 
Consider any node $(\cdfunc,S)\neq (f,V)$ of the canonical decomposition tree,
and let $(\cdfunc_1,S_1)$ be the parent of $(\cdfunc,S)$. 
Applying part (a) of Lemma~\ref{decomptree} to $(\cdfunc_1,S_1)$ taking
$T=S$, since $\defc_{\cdfunc_1,\gs[{S_1}]}(S)>0$, we obtain that there is some
small-$(f,G)$-cut $A\sse V$ such that $S=A\cap S_1$. Thus, 
$\dt_{\gs[{S_1}]}(S)\sse\dt_G(A)\sse\smcuts$.
Now for any leaf-set $V_i$, every edge in $\dt_G(V_i)$
belongs to $\dt_{\gs[{S_1}]}(S)$, where $(\cdfunc,S)$ is some ancestor of $(f_i,V_i)$ and
$(\cdfunc_1,S_1)$ is its parent.
So we obtain that $\dt_G(V_i)\sse\smcuts$ for every leaf-set $V_i$. 

Now, consider some leaf-set $V_i$, and consider any two components 
$\spart_j,\spart_{j'}$, where $\spart_j\sse V_i$. We have to show that $\spart_j$ and 
$\spart_{j'}$ are non-separable iff $\spart_{j'}\sse V_i$. Suppose first that
$\spart_{j'}\nsubseteq V_i$. Let $\spart_{j'}\sse V_q$ for some leaf-set $V_q\neq V_i$. 
Let $(\cdfunc,S)$ be the least common ancestor of $(f_i,V_i)$ and $(f_q,V_q)$, and
let $(\cdfunc_1,\scut)$, $(\cdfunc_2,S-\scut)$ be the children of $(\cdfunc,S)$. Clearly,
$\scut$ separates $V_i$ and $V_q$.
We have that $0<\defc_{\cdfunc,\gs}(\scut)=\max\,\{\defc_{\cdfunc,\gs}(T): T\sse S\}$.
So for any $s\in \scut$, $t\in S-\scut$, by Lemma~\ref{decomptree} (c),
we may assume that $\scut=T\cap S$, where $\defc_{f,G}(T)=\defc_{\cdfunc,\gs}(\scut)$.
It follows that $T$ is a small-$(f,G)$-cut that separates $V_i$ and $V_q$, and hence
also separates $\spart_j$ and $\spart_{j'}$.

Now suppose that $\spart_{j'}\sse V_i$. 
Suppose there is some small-$(f,G)$-cut $T$ that separates $\spart_j$ and $\spart_{j'}$.
Then taking any node $s\in\spart_j$ and any $t\in\spart_{j'}$, by Lemma~\ref{decomptree}
(c), there is some small-$(f_i,\gs[{V_i}])$-cut $T'\subsetneq V_i$ separating $s$ and
$t$. This yields a contradiction since $(f_i,V_i)$ is a leaf of the canonical
decomposition tree. 
\end{proof}

\begin{lemma} \label{decomp-fns}
Let $(f_1,V_1),\ldots,(f_k,V_k)$ be the leaves of the canonical decomposition tree.
Consider $i\in[k]$, and
let $\scol_i:=\bigl\{\bigcup_{j\in I}V_j: I\sse[k]-\{i\}\bigr\}$.
Define $\nf_i:2^{V_i}\mapsto\Z$ as follows:
\begin{equation*}
\nf_i(\es)=\nf_i(V_i):=0, \qquad
\nf_i(S):=\max\,\bigl\{f(S\cup X)-|\dt_{\smcuts}(S\cup X)|: X\in\scol_i\bigr\} 
\quad \forall \es\neq S\subsetneq V_i.
\end{equation*}
Then $\nf_i$ coincides with $f_i$. 
the function obtained from the canonical decomposition, for all $i\in[k]$.
\end{lemma}

\begin{proof}
Let $(\cdfunc_1,S_1),\ldots,(\cdfunc_\ell,S_\ell)=(f,V)$ be the nodes
of the canonical decomposition tree on the path from $(f_i,V_i)$ to the root, with
$(\cdfunc_1,S_1)$ being the parent of $(f_i,V_i)$. 
Consider any $\es\neq T\subsetneq V_i$. 
We first show that 
$\nf_i(T)\geq f_i(T)$, or equivalently $\defc_{\nf_i,\gs[V_i]}(T)\geq\defc_{f_i,\gs[V_i]}(T)$.
This is because, by Lemma~\ref{decomptree} (a), one can infer that $f_i$ is defined
similarly to $\nf_i$, except that we restrict the sets $X\in\scol_i$ over which the
maximum is taken. Recall that for $I\sse[\ell]$, we
define $R_I:=\bigcup_{j\in I}(S_j-S_{j-1})$, and $R_j$ denotes $R_{\{j\}}:=S_j-S_{j-1}$,
for $j\in[\ell]$. Observe that for any $j\in[\ell]$, we have that $S_j-S_{j-1}$ is a union
of some leaf-sets other than $V_i$. Therefore, $S_j-S_{j-1}\in\scol_i$ for every
$j\in[\ell]$, and hence $R_I\in\scol_i$ for all $I\sse[\ell]$.
We have
\begin{equation*}
\defc_{f_i,\gs[V_i]}(T)=\max\,\bigl\{\defc_{f,G}(T\cup R_I): I\sse[\ell]\bigr\}
\leq\max\,\bigl\{\defc_{f,G}(T\cup X): X\in\scol_i\bigr\}
=\defc_{\nf_i,\gs[V_i]}(T).
\end{equation*}
The first inequality is because $R_I\in\scol_i$ for every $I\sse[\ell]$, and the last
equation follows because $\dt_G(T\cup X)$ is the union of two disjoint sets,
$\dt_{\smcuts}(T\cup X)$ and $\dt_{\gs[V_i]}(T)$.

We now argue that $\nf_i(T)\leq f_i(T)$. 
Let $\defc_{\nf_i,\gs[V_i]}(T)=\defc_{f,G}(T\cup X)$, where $X\in\scol_i$. Let $T_1=T\cup X$.
In the proof of part (c) of Lemma~\ref{decomptree}, we argued that
$\defc_{f,G}(T_1)\leq\defc_{f_i,\gs[V_i]}(T_1\cap V_i)$ (see \eqref{decomp-inv}). But
$T_1\cap V_i=T$, so $\defc_{\nf_i,\gs[V_i]}(T)\leq\defc_{f_i,\gs[V_i]}(T)$.
\end{proof}

\subsubsection{Further remarks and discussion} \label{altdecomp-remk}

\paragraph*{Alternate ``direct'' proof.}
The reader may wonder whether one can directly argue that the $(\nf_i,V_i)$
tuples specified succinctly in Lemmas~\ref{decomp-partn}
and~\ref{decomp-fns} satisfy the properties stated in Theorem~\ref{redn-cor}. It is
easy to show that if $F$ is feasible for $\crfndp[{(f,G)}]$, then we
must have $F\supseteq\dt_G(V_i)$ and that $F(V_i)$ is a feasible solution to
$\fndp[{(\nf_i,\gs[V_i])}]$, for all $i\in[k]$. This is simply because, 
for any $i\in[k]$, any set $\es\neq S\subsetneq V_i$, and any $X\in\scol_i$, we 
have $|\dt_F(S\cup X)|\geq\min\bigl\{f(S\cup X),|\dt_G(S\cup X)|\bigr\}=f(S\cup X)$ 
(the last equality is because $\dt_G(S\cup X)\not\subseteq\smcuts$), which
implies that 
$|\dt_{F(V_i)}(S)|\geq f(S\cup X)-|\dt_{\smcuts}(S\cup X)|$. 

The other direction is more tricky.
Any small-$(f,G)$-cut $\scut$ must be a union of leaf-sets, because if 
$\scut\cap V_i, V_i-\scut$ are both non-empty for some leaf-set $V_i$, then
Lemma~\ref{decomptree} (c) implies that there is some small-$(f_i,\gs[{V_i}])$-cut, which
contradicts that $(f_i,V_i)$ is a leaf of the canonical decomposition tree.
Since $F\supseteq\bigcup_{i\in[k]}\dt_G(V_i)$, we can focus on sets that are not
a union of leaf-sets, and hence are not small-$(f,G)$-cuts.
Recall that we say that two sets $A, B$ cross if $A-B, A\cap B$ and $B-A$ are all
non-empty. Suppose $S\sse V$ crosses some $r$ leaf sets. 
Similar to the proof of Lemma~\ref{main-redn}, one potential approach would be to 
utilize the weak supermodularity of $\defc_{f,F}$ 
to move to sets that cross fewer leaf sets, and thereby use induction on $r$ to
establish that $\defc_{f,F}(S)\leq 0$. 
{\em There are some subtle issues that arise here.}
 
The base cases, $r=0$ or $r=1$ are easy enough to argue, using 
the fact that $F(V_i)$ is a feasible $\fndp[{(\nf_i,\gs[V_i])}]$-solution for all
$i\in[k]$. 
So suppose $r\geq 2$, 
A natural idea 
would be to pick some leaf set that $S$ crosses, say $V_i$, and apply weak supermodularity
to $S$ and $V_i$, since we know that all four sets $S\cap V_i, S\cup V_i, S-V_i, V_i-S$
cross fewer than $r$ leaf sets, and none of these sets is a union of leaf-sets and hence a 
small-$(f,G)$-cut. This would yield, for  
$\{X,Y\}=\{S\cap V_i, S\cup V_i\}$ or $\{X,Y\}=\{S-V_i, V_i-S\}$, that
$\defc_{f,F}(S)+\defc_{f,F}(V_i)\leq\defc_{f,F}(X)+\defc_{f,F}(Y)\leq 0$. 
However, since $V_i$ need not be a small-$(f,G)$-cut, this does not let us infer anything
about $\defc_{f,F}(S)$. 
To fix this, we would like to pick some {\em small-$(f,G)$-cut} $T\supseteq V_i$, such
that all the sets $S\cap T, S\cup T, S-T, T-S$ cross fewer than $r$ leaf sets (so that 
one can apply the induction hypothesis to these four sets). 
In particular, we need to avoid the situation where $T$ contains all the same leaf-sets
that $S$ crosses, because then $S\cap T$ and $T-S$ would also cross all of these leaf
sets. 

Fortunately, the proof of Lemma~\ref{decomp-partn} shows that any two leaf-sets can be
separated by some small-$(f,G)$-cut. (This utilizes insights from
Lemma~\ref{decomptree}; we do not know of any other way of proving this.)
This ``leaf-set separation'' property supplies
the key missing ingredient needed to make the above induction proof work. 
Let $V_i$ and $V_j$ be two leaf sets that $S$ crosses. Now let $T$ be a small-$(f,G)$-cut
separating $V_i$ and $V_j$. Now, all four sets $S\cap T, S\cup T, S-T, T-S$ cross 
fewer than $r$ leaf sets, so by weak-supermodularity and the induction hypothesis, we
obtain that $\defc_{f,F}(S)+\defc_{f,F}(T)\leq 0$, which implies that 
$\defc_{f,F}(S)\leq 0$, since $T$ is a small-$(f,G)$-cut.

\paragraph*{Alternate decomposition.}
The succinct description of the $\nf_i$ functions in Lemma~\ref{decomp-fns} suggests
another decomposition into $\fndp$ instances, where the cut-requirement functions are now
defined directly on the components of $G-\smcuts$ in a manner similar to the $\nf_i$
functions. Recall that $\spart_1,\ldots,\spart_p$ are the components of $G-\smcuts$.   
To elaborate, for $j\in[p]$, define $\newf_j:2^{\spart_j}\mapsto\Z$ as follows:
$\newf_j(\es)=\newf_j(\spart_j):=0$, and
\begin{equation*}
\newf_j(S):=\max\,\Bigl\{f(S\cup X)-|\dt_{\smcuts}(S\cup X)|: X=\bigcup_{j'\in I}\spart_{j'},\ I\sse[p]-\{j\}\Bigr\}
\quad \forall \es\neq S\subsetneq\spart_j.
\end{equation*}
We can show that the $(\newf_1,\spart_1),\ldots,(\newf_p,\spart_p)$ tuples satisfy the
properties stated in Theorem~\ref{redn-cor}. Consider a component $\spart_j$ and leaf-set
$V_i$ containing $\spart_j$. Let $\eqclass_i\sse[p]$ be such that 
$V_i=\bigcup_{j'\in\eqclass_i}\spart_{j'}$. To gain intuition, consider any
$S\sse\spart_j$ and consider the contraints imposed by $\nf_i$ on $\dt_{\gs[{V_i}]}(S)$.
Observe that $\dt_{\gs[{V_i}]}(\spart_{j'})=\es$ for any $j'\in\eqclass_i$, since
$\dt_G(\spart_{j'})\sse\smcuts$ and there
are no edges of $\smcuts$ internal to $V_i$.
So for any $Y=\bigcup_{j'\in J}\spart_{j'}$, where $J\sse\eqclass_i-\{j\}$, we can infer
that $\nf_i$ imposes that we pick $\nf_i(S\cup Y)$ edges from $\dt_{\gs[{\spart_j}]}(S)$. 
We argue that $\newf_j$ is simply the function that aggregates these
constraints. 
More precisely, we 
show that 
\begin{equation}
\newf_j(S)=\max\,\Bigl\{\nf_i(S\cup Y): Y=\bigcup_{j'\in J}\spart_{j'},\ J\sse\eqclass_i-\{j\}\Bigr\}
\quad \forall \es\neq S\subsetneq\spart_j. \label{newdecomp}
\end{equation}
This also implies that the $(\newf_j,\spart_j)$ tuples satisfy the properties in 
Theorem~\ref{redn-cor}. 

As in the proof of Lemma~\ref{decomp-fns}, it is easy to see that for any $S\sse\spart_j$,
we have $\newf_j(S)$ is at least the RHS of \eqref{newdecomp}.
Consider any $Y=\bigcup_{j'\in J}\spart_{j'}$, $J\sse\eqclass_i-\{j\}$.
(Recall that $\scol_i=\bigl\{\bigcup_{j\in I}V_j: I\sse[k]-\{i\}\bigr\}$, where
$V_1,\ldots,V_k$ are the leaf sets of the canonical decomposition.)
If $\nf_i(S\cup Y)=f(S\cup Y\cup X)-|\dt_{\smcuts}(S\cup Y\cup X)|$, where $X\in\scol_i$,
then we can see that $Y\cup X$ is of the form $\bigcup_{j'\in I}\spart_{j'}$ for some
$I\sse[p]-\{j\}$, and so 
$\newf_j(S)\geq f(S\cup Y\cup X)-|\dt_{\smcuts}(S\cup Y\cup X)|=\nf_i(S\cup Y)$.
To show that $\newf_j(S)$ is at most the RHS of \eqref{newdecomp},
suppose $\newf_j(S)=f(S\cup X)-|\dt_{\smcuts}(S\cup X)|$, where $X=\bigcup_{j'\in I}\spart_{j'}$,
$I\sse[p]-\{j\}$. Let $T=S\cup X$.
So $\defc_{\newf_j,\gs[{V_i}]}(S)=\defc_{f,G}(T)$. As in the proof
of Lemma~\ref{decomp-fns}, by Lemma~\ref{decomptree} (c), we have 
$\defc_{f,G}(T)\leq\defc_{f_i,\gs[{V_i}]}(T\cap V_i)$. But $T\cap V_i$ is of the form
$S\cup Y$, where $Y=\bigcup_{j'\in J}\spart_{j'}$, $J\sse\eqclass_i-\{j\}$, so we obtain
that $\defc_{\newf_j,\gs[{V_i}]}(S)\leq\defc_{f_i,\gs[{V_i}]}(S\cup Y)$, which translates
to $\newf_j(S)\leq f_i(S\cup Y)=\nf_i(S\cup Y)$.

\section{Hardness result} \label{hardness}
We now show that 
cut-relative \sndp and path-relative \sndp are \apxhard, even when the base \sndp instance 
involves only one $s$-$t$ pair, 
by showing that the $k$-edge connected subgraph (\kecss) problem can be cast as a special
case of these problems. (This hints at the surprising 
amount of modeling power of \{cut, path\}- relative \sndp.)
Previously, even \nphard{}ness of path-relative \sndp in the $s$-$t$ case was not known.

Let $\bigl(G=(V,E),\{c_e\}_{e\in E}\bigr)$ be an instance of \kecss. We may assume that
$G$ is at least $k$-edge connected, i.e., $|\dt_G(S)|\geq k$ for all 
$\es\neq S\subsetneq V$, as 
otherwise there is no feasible solution (and this can be efficiently detected).

Consider the following path-relative \sndp instance. We add a source $s$ and sink $t$, and
edges $sv, vt$ of cost $0$, for all $v\in V$. Let $G'=(V',E')$ denote this graph. The
base-\sndp instance asks for $(k+n)$-edge connectivity between $s$ and $t$, where
$n=|V|$. That is, the base cut-requirement function $f$ is given by $f(S')=k+n$ if 
$S'\sse V'$ is an $s$-$t$ cut, and $f(S')=0$ otherwise. 
We show that path-relative \sndp and cut-relative \sndp are actually
equivalent in this case, and that they encode \kecss on the graph $G$.

\begin{theorem} \label{apxhard}
$H\sse E'$ is a feasible solution to the above \{cut, path\}-relative \sndp instance $\iff$
$H\supseteq\dt_{G'}(s)\cup\dt_{G'}(t)$ and $H(V)$ is a feasible \kecss-solution for $G$.
Hence, this \{cut,path\}-relative \sndp instance is the same as \kecss on the graph $G$.
\end{theorem}

\begin{proof} 
Recall that we are assuming that $G$ is (at least) $k$-edge connected.
The following observation will be handy. The only small cuts in $G'$ are $\{s\}$ and
$\{t\}$ (and their complements). For any other $s$-$t$ cut $S'\sse V'$, 
taking $S=S'\cap V$, we have that $\es\neq S\subsetneq V$. So 
we have $|\dt_{G'}(S')|=|\dt_G(S)|+|S|+|V-S|\geq k+n$.

\medskip\noindent
\underline{\bf \boldmath $\Leftarrow$ direction.}
We argue that $H$ is feasible for cut-relative \sndp, which also
implies that it is feasible for path-relative \sndp. 
Clearly, $H$ covers both small cuts.
For any other $s$-$t$ cut $S'\sse V'$, taking $S=S'\cap V$, we have 
$|\dt_H(S')|\geq|\dt_{H(V)}(S)|+|S|+|V-S|\geq k+n$, since $H(V)$ is a feasible
\kecss-solution for $G$. 

\medskip\noindent
\underline{\bf\boldmath $\Rightarrow$ direction.}
Suppose that $H$ is feasible for path-relative \sndp. If there is some edge $sv\notin H$,
then consider the fault-set $F=\dt_{G'}(s)-\{sv\}$. Then $H-F$ clearly has no $s$-$t$
path, since $\dt_{H-F}(s)=\es$. But $G'-F$ has an $s$-$t$ path $sv,vt$. So we must have
$H\supseteq\dt_{G'}(s)$. A similar argument shows that $H\supseteq\dt_{G'}(t)$. 

We next argue that $|\dt_H(S')|\geq n+k$ for every $s$-$t$ cut $S'$ for which 
$\es\neq S:=S'\cap V\subsetneq V$. This is equivalent to showing that 
$|\dt_{H(V)}(S)|\geq k$, since $|\dt_H(S')\cap\dt_{G'}(s)|+|\dt_H(S')\cap\dt_{G'}(t)|=n$.
Assume that $s\in S'$ without loss of generality. Suppose $|\dt_{H(V)}(S)|<k$. Since
$|\dt_G(S)|\geq k$, there is some edge $uv\in\dt_G(S)-\dt_{H(V)}(S)$. Suppose $u\in S$. 
Consider the fault-set $F=\dt_{H(V)}(S)\cup\{sw: w\in V-S\}\cup\{wt: w\in S\}$.
Then $|F|<k+n$ and $H-F$ has no $s$-$t$ path, since we have constructed $F$ to ensure that
$\dt_{H-F}(S')=\es$. However, $G-F$ has an $s$-$t$ path: $su,uv,vt$. This contradicts that
$H$ is feasible for path-relative \sndp. 

So we have shown both that $H(V)$ is feasible \kecss solution for $G$, and that $H$ is
feasible for cut-relative \sndp.

Since the cost of the edges in $E'-E$ is $0$, the above \{cut, path\}-relative \sndp
instance is exactly the same as \kecss on the graph $G$. 
\end{proof}

\bibliography{rsndp}

\appendix

\section{Non-equivalence of path-relative \sndp and cut-relative \sndp}
\label{append-nonequiv} \label{nonequiv}

We consider the same graph and base-SNDP problem as in Fig.~\ref{nolaminar-fig}.
The graph is reproduced below for convenience, 
and recall that the base cut-requirement function models $s$-$t$ $2$-edge connectivity: 
so $f(S)=2$ if $S$ is an $s$-$t$ cut, and $f(S)=0$ otherwise.  

\begin{figure}[ht!]
\centering
\includegraphics[width=2.25in]{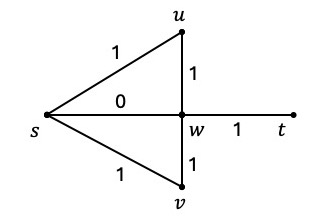}
\caption{An instance where path-relative \sndp and cut-relative \sndp are not
  equivalent. The number labeling an edge gives the cost of the edge.}
\label{nonequiv-fig}
\end{figure}

A feasible solution to path-relative \sndp is given by the edge-set $H_1=\{su,uw,sw,wt\}$,
which has cost $3$. (This is in fact an optimal solution to path-relative \sndp.)

However, {\em $H_1$ is not feasible for cut-relative \sndp}. This is because, we require
$\dt_H(\{s,u,w\})\geq\min\{2,3\}=2$ in any feasible \crsndp solution $H$. Moreover, one
can infer that any feasible \crsndp-solution must include all but one of the edges from
$\{su,uw,sw,sv,vw\}$. This is because for any $F\sse\{su,uw,sw,sv,vw\}$ with $|F|=2$, one
can verify that there is an $s$-$t$ cut $S$ such that $|\dt_G(S)|=3$ and $F\sse\dt_G(S)$; 
so $|\dt_{G-F}(S)|\leq 1$, which implies that there is no feasible solution that excludes
$F$. A feasible \crsndp-solution must also include the edge $wt$, so the optimal value for
cut-relative \sndp is $4$.

\smallskip
While cut-relative \sndp and path-relative \sndp are not in general equivalent, they are 
closely related. Dinitz et al.~\cite{DinitzKKN23} showed (see Lemma A.1
in~\cite{DinitzKKN23}) that one can give a cut-based formulation for path-relative \sndp,
where $H$ is feasible iff  
$|\dt_H(S)|\geq\min\bigl\{f(S),|\dt_G(S)|\}$ for every $S\sse V$ {\em such that $\gs$ and
$\gs[V'-S]$ are connected graphs}, where $V'\sse V$ is the vertex-set of the component 
containing $S$. That is, we require the cut-constraints of \crsndp to hold for a suitable
collection of node-sets (as opposed to all node-sets in \crsndp). This immediately implies
that \crsndp and path-relative \sndp are equivalent when the underlying graph is a
(capacitated) complete graph. Also, as mentioned earlier, path-relative \kecss and
cut-relative \kecss are equivalent.

\begin{lemma}[Follows from~\cite{DinitzKK22}] \label{kecss-equiv}
Path-relative \kecss and cut-relative \kecss are equivalent.
\end{lemma}

\begin{proof}
This was shown by~\cite{DinitzKK22}. We include a proof for completeness.
We have already seen that a feasible solution to \crsndp is also feasible for
path-relative \sndp. Now suppose that $H\sse E$ is a feasible solution to
path-relative \kecss. 
Suppose $\dt_H(S)<\min\bigl\{k,|\dt_G(S)|\bigr\}$ for some set $\es\neq S\subsetneq V$.   
Consider $F=|\dt_H(S)|$, and let $e=uv\in\dt_{G-F}(S)$. Then, $|F|<k$ and $G-F$ has a
$u$-$v$ path, but $H-F$ has no $u$-$v$ path, which contradicts that $H$ is feasible for 
path-relative \kecss.  
\end{proof}

\section{Proof of part (b) of Theorem~\ref{main-redn}} \label{append-decomp}

We mimic the proof of part (a). Recall that $Z=\dt_G(\scut)$, 
$(f_1,G_1,F_1)=\bigl(f_{\scut},\gs[\scut],F(\scut)\bigr)$ and
$(f_2,G_2,F_2)=\bigl(f_{\scutc},\gs[\scutc],F(\scutc)\bigr)$.

We extend the notion of deficiency to fractional solutions, or equivalently weighted
graphs. Given graph $D=(V_D,E_D)$, weights $w\in\R_+^{E_D}$ and cut-requirement function
$h:2^{V_D}\mapsto\Z$, define 
the {\em weighted-deficiency} of a set $S\sse V$ with respect to $(h,w)$
to be $\wdef_{h,w}(S):=h(S)-w\bigl(\dt_D(S)\bigr)$.
If $w_e=1$ for all $e\in Z$, then as in part (b) of Lemma~\ref{split-wsupm}, we
have
$\wdef_{f_1,w^{\scut}}(T)=\max\bigl\{\wdef_{f,w}(T),\wdef_{f,w}(T\cup\scutc)\bigr\}$
for $T\sse\scut$. Similarly, by considering $f_2$, the restriction of $f$ to $\scutc$, 
we have
$\wdef_{f_2,w^{\scutc}}(T)=\max\bigl\{\wdef_{f,w}(T),\wdef_{f,w}(T\cup\scut)\bigr\}$
for $T\sse\scutc$.

Let $x\in\R_+^E$. Suppose first that $x$ is feasible to $\lpname[f,G]$. 
Since $\scut$ is a small-$(f,G)$-cut, we must have $x_e=1$ for all $e\in Z$. Consider any
$\es\neq T\subsetneq\scut$. 
We have $x\bigl(\dt_G(T)\bigr)\geq\min\bigl\{f(T),|\dt_G(T)|\bigr\}$, so since $x_e=1$ for
all $e\in Z$, we have 
$x^{\scut}\bigl(\dt_{G_1}(T)\bigr)\geq\min\bigl\{f(T)-|\dt_Z(T)|,|\dt_{G_1}(T)|\bigr\}$.
We also have
$x\bigl(\dt_G(T\cup\scutc)\bigr)\geq\min\bigl\{f(T\cup\scutc),|\dt_G(T\cup\scutc)|\bigr\}$,
and so
$x^{\scut}\bigl(\dt_{G_1}(T)\bigr)\geq\min\bigl\{f(T\cup\scutc)-|\dt_Z(T\cup\scutc)|,|\dt_{G_1}(T)|\bigr\}$. 
Combining these, we obtain that
$x^{\scut}\bigl(\dt_{G_1}(T)\bigr)\geq\min\bigl\{f_1(T),|\dt_{G_1}(T)|\bigr\}$. 
Thus shows that $x^{\scut}$ is feasible to $\lpname[{(f_1,G_1)}]$.

Again, a symmetric argument interchanging the roles of $\scut, f_1, F_1, G_1$ with
$\scutc, f_2, F_2, G_2$ shows that
$x^{\scutc}$ is a feasible solution to $\lpname[{(f_2,G_2)}]$.

\medskip
Now suppose that $x_e=1$ for all $e\in Z$, 
$x^{\scut}$ is feasible to $\lpname[{(f_1,G_1)}]$, and
$x^{\scutc}$ is feasible to $\lpname[{(f_2,G_2)}]$.
Consider any $\es\neq T\subsetneq V$. 
Suppose $T\sse\scut$. If $T$ is a small-$(f_1,G_1)$-cut, then $x_e=1$ for all
$e\in\dt_{G_1}(T)$; since $x_e=1$ for all $e\in Z$, this implies that
$x_e=1$ for all $e\in\dt_G(T)$. Otherwise, we have $\wdef_{f_1,x^{\scut}}(T)\leq 0$,
since $x^{\scut}$ is feasible to $\lpname[{(f_1,G_1)}]$, and
so $\wdef_{f,x}(T)\leq\wdef_{f_1,x^{\scut}}(T)\leq 0$. This shows that the
constraints for sets $T\sse\scut$ in $\lpname[f,G]$ is satisfied. A symmetric argument shows
that the same holds for sets $T\sse\scutc$.

Now consider $T$ such that $T\cap\scut$ and $T\cap\scutc$ are both non-empty.
As in the proof of part (a), since $\wdef_{f,x}$ is weakly supermodular and symmetric, we
can find sets $X\sse\scut$, $Y\sse\scutc$ such that
$\wdef_{f,x}(T)+\wdef_{f,X}(\scut)\leq\wdef_{f,x}(X)+\wdef_{f,x}(Y)$ and
$\dt_G(T)\sse\dt_G(X)\cup\dt_G(Y)\cup Z$. 
We showed above that the constraints for $X$ and $Y$ in $\lpname[f,G]$ are
satisfied, which can be stated equivalently as 
$\wdef_{f,x}(S)\leq\max\bigl\{0,\defc_{f,G}(S)\bigr\}$ for $S\in\{X,Y\}$.
Also, $\wdef_{f,x}(\scut)=\defc_{f,G}(\scut)>0$. 
So for $S\in\{X,Y\}$, we have
\begin{equation}
\wdef_{f,x}(S)\leq\max\bigl\{0,\defc_{f,G}(S)\bigr\}
\leq\max\bigl\{0,\defc_{f,G}(\scut)\bigr\}=\defc_{f,G}(\scut)=\wdef_{f,x}(\scut)
\label{mainredn-ineq1}
\end{equation}
where the second inequality is because $\scut$ is a maximum-deficiency set.
If both $X$ and $Y$ are small-$(f,G)$-cuts, then we have $x_e=1$ for all
$e\in\dt_G(X)\cup\dt_G(Y)\cup Z$, and hence $x_e=1$ for all $e\in\dt_G(T)$.
Otherwise, $\wdef_{f,x}(T)+\wdef_{f,X}(\scut)\leq\wdef_{f,x}(X)+\wdef_{f,x}(Y)$, coupled
with \eqref{mainredn-ineq1} shows that
$\wdef_{f,x}(T)\leq\min\bigl\{\wdef_{f,x}(X),\wdef_{f,x}(Y)\bigr\}$, which is at most $0$
since at least one of $X$, $Y$ is not a small-$(f,G)$-cut.
This shows that $x$ is feasible to $\lpname[(f,G)]$. \hfill \qed

\end{document}